\documentclass[a4paper,12pt]{amsart}

\usepackage{color,graphicx}

\usepackage{amscd,amssymb,amsthm,latexsym,stmaryrd,xypic,amsmath}

\xyoption{all}

\usepackage{hyperref}
\hypersetup{
hyperindex=true, 
colorlinks=true, 
linkcolor=blue, 
linktocpage=true,           
filecolor=red,
citecolor=red, 
breaklinks=false, 
urlcolor= blue, 
bookmarks=true, 
bookmarksopen=false,
}
  
\numberwithin{equation}{section}
\theoremstyle{plain} %
\newtheorem{thm}[equation]{Theorem}
\newtheorem{prop}[equation]{Proposition}
\newtheorem{coro}[equation]{Corollary}
\newtheorem{lem}[equation]{Lemma}
\theoremstyle{definition} %
\newtheorem{ex}[equation]{Example}

\newtheorem{rem}[equation]{Remark}

\renewenvironment{proof}{{\it Proof. }}{
  \qed\endtrivlist
  \par}

\newcommand{\OO}{{\mathcal{O}}}
\newcommand{\pp}{\mathfrak{p}}

\newcommand{\PP}{\mathfrak{P}}
\newcommand{\PPP}{{\scriptscriptstyle\mathfrak{P}}}

\newcommand{\qq}{\mathbb{Q}}
\newcommand{\rr}{\mathbb{R}}
\newcommand{\zz}{\mathbb{Z}}
\newcommand{\ff}{\mathbb{F}}

\newcommand{\ds}{\displaystyle}

\newcommand{\Gal}{\mathrm{Gal}}

\newcommand{\Id}{\mathrm{Id}}

\newcommand{\Tr}{\mathrm{Tr}}
\newcommand{\Trd}{\mathrm{Trd}}
\newcommand{\codim}{\mathrm{codim}}
\def\ov#1{{\overline{#1}}}

\def\zn#1{\zz/#1\zz}

\renewcommand{\to}{\longrightarrow}
\renewcommand{\mapsto}{\longmapsto}

\newcommand{\funcisop}[4]{\begin{aligned}\newline #1&\overset{\sim}{\longrightarrow} #2 \cr\newline #3 &\longmapsto #4.\end{aligned}}
\newcommand{\funcv}[4]{\begin{aligned}\newline #1&\longrightarrow #2 \cr\newline #3 &\longmapsto #4,\end{aligned}}
\newcommand{\nfunc}[5]{#1\colon\begin{aligned}\newline #2&\longrightarrow #3 \cr\newline #4 &\longmapsto #5\end{aligned}}

\newcommand{\nfuncp}[5]{#1\colon\begin{aligned}\newline #2&\longrightarrow #3 \cr\newline #4 &\longmapsto #5.\end{aligned}}
\newcommand{\nfuncv}[5]{#1\colon\begin{aligned}\newline #2&\longrightarrow #3 \cr\newline #4 &\longmapsto #5,\end{aligned}}

\let\cdotorg\cdot \def\cdot{{\cdotorg}}

\newcommand{\mJ}{{\mathcal{J}}}
\newcommand{\mP}{{\mathcal{P}}}
\newcommand{\mPP}{{\scriptscriptstyle\mathcal{P}}}
\usepackage{verbatim}
\usepackage{tikz-cd}

\title[A Quadratic Form Approach to Construction A]{A Quadratic Form Approach to Construction A of Lattices over Cyclic Algebras}
\author{Gr\'{e}gory Berhuy, Fr\'{e}d\'{e}rique Oggier}
\date{\today}

\newcommand{\dl}{\sharp}

\begin{document}

\maketitle

\begin{abstract}
We propose a construction of lattices from (skew-) polynomial codes, by endowing quotients of some ideals in both number fields and cyclic algebras with a suitable trace form. We give criteria for unimodularity. This yields integral and unimodular lattices with a multiplicative structure. Examples are provided.
\end{abstract}

%
%
%

\tableofcontents

\section{Introduction}

A classical theory connects linear codes, linear subspaces of $\ff_q^n$, and Euclidean lattices (see e.g. \cite{splag} for a wealth of references on the so-called Constructions A,B,C,D). Construction A provides a framework to construct lattices from linear codes, under which several correspondences are known, such as between the dual code and the dual lattice, self-dual codes and unimodular lattices, weight enumerators and theta series. It was for example used to obtain extremal modular lattices \cite{bachoc}. Applications coming from cryptography motivated to revisit Construction A by adding a multiplicative structure (see e.g. \cite{belfiore} which contains examples of multiplicative constructions over number fields). Some examples over cyclic algebras are presented in \cite{ducoat} (with applications to wiretap coding), where skew-polynomial codes \cite{boucher-ulmer} are replacing linear codes.
While both \cite{belfiore,ducoat} provided examples of multiplicative constructions, they also both left open the question of duality for these settings, which we address in this paper, via the theory of quadratic forms.

%
%
%
\section{From Codes to Lattices}\label{cod-lat}

\subsection{Generalities about Lattices}
Let $(E,b)$ be a Euclidean space of dimension $n$ with $b$ a positive definite symmetric bilinear form, let $M$ be a full lattice of $E$, that is a subgroup generated by a basis of $E$ (endowed with the induced bilinear form). In particular $M$ has rank $n$.

If $(e_1,\ldots,e_n)$ is a $\zz$-basis of $M$, $\det(M)$ will denote the determinant of $B = (b(e_i,e_j))_{i,j}$. This determinant is positive, since $b$ is positive definite.

If $N$ is a full sublattice of $M$, it is known that the index $[M:N]$ of $N$ in $M$ is finite and $\det(N)=[M:N]^2\det(M)$.

 dual of $M$ is the lattice $$M^\dl=\{x\in E\mid b(x,y)\in\zz \mbox{ for all }y\in M\}.$$
It has a $\zz$-basis $(e^\dl_1,\ldots,e^\dl_n)$ defined by $b(e^\dl_i,e_j)=1$ if $i=j$ and $0$ else. Thus $B^\dl = (b(e^\dl_i,e^\dl_j))_{i,j} = B^{-1}$ and  $\det(M^\dl)=\det(M)^{-1}$.
If $M = M^\dl$, the lattice $M$ is said to be unimodular.

\subsection{Quotients of Lattices}
Fix a prime number $p$. Assume that $N$ is a sublattice of $M$ such that $pM\subset N\subset M$, and let $\pi: M\to M/N$ be the canonical projection.

If $x\in M$, we will denote by $[x]_N$ its class in $M/N$. In the case where $M$ is a ring and $N$ is an ideal generated by a single element $a$, we will simply denote it by $[x]_a$.

 Since $pM\subset N$, $M/N$ has a natural structure of an $\ff_p$-vector space, given by $$[m]_p\cdot [x]_N=[m\cdot x]_N, \ \mbox{ for all }m\in\zz, x\in M.$$ Since $M$ has rank $n$, $\dim_{\ff_p}(M/N)\leq n$.


Consider now $M$ an integral lattice, meaning that $M\subset M^\dl$, which is equivalent to say that $b(x,y)\in \zz$ for all $x,y\in M$. As above, $N$ is a sublattice of $M$ such that $pM\subset N\subset M$.
Assume also that $b(x,y)\in p\zz$ for all $x\in M$ and $y\in N$.
Then $b$ induces  on $M$ a symmetric $\zz$-bilinear form $b:M\times M\to\zz$, which in turn induces an $\ff_p$-bilinear form $$\nfuncp{\ov{b}}{M/N\times M/N}{\ff_p}{([x]_N,[y]_N)}{[b(x,y)]_p}$$
For $u,v \in M$ of the form $u=x+n,~v=y+n'$, $n,n'\in N$, we have $([u]_N,[v]_N) \mapsto [b(u,v)]_p$ and since $[b(x+n,y+n')]_p = [b(x,y)]_p$ using that $b(x,y)\in p\zz$ for all $x\in M$ and $y\in N$, $\bar{b}$ is well-defined. It is however not necessarily nondegenerate.

In fact, we have the following lemma.

\begin{lem}\label{nondeg}
In the situation above, the radical of the $\ff_p$-bilinear map $\ov{b}:M/N\times M/N\rightarrow \ff_p$ is $(pM^\dl\cap M)/N$.
In particular, $\ov{b}$ is nondegenerate if and only if $pM^\dl\cap M=N.$
\end{lem}

\begin{proof}
Let $x\in M$. Then $\ov{b}([x]_N,[y]_N)=[0]_p$ for all $[y]_N\in M/N$ if and only if $b(x,y)\in p\zz$ for all $y\in M$.
If $x\in pM^\dl\cap M$, then $x=px'$ for some $x'\in M^\dl$, and $b(x,y)=pb(x',y)\in p\zz$ for all $y\in M$. Conversely, if 
$b(x,y)\in p\zz$ for all $y\in M$, then $\dfrac{x}{p}\in M^\dl$, and $x=p\dfrac{x}{p}\in pM^\dl\cap M$.

Hence, the radical of $\ov{b}$ is $(pM^\dl\cap M)/N$, as required. 
\end{proof}

The case $N=pM$ is the easiest to handle. It automatically satisfies both $pM\subset N\subset M$ and $b(x,y)\in p\zz$ for all $x\in M$ and $y\in N$.

\begin{ex}\label{ex-pm}
Assume that $N=pM$. Let $(e_1,\ldots,e_n)$ be a $\zz$-basis of $M$, and let $B=(b(e_i,e_j))_{i,j}$.

We have a canonical isomorphism of $\ff_p$-vector spaces $ M\otimes_\zz \ff_p\simeq M/N$, which sends $x\otimes [1]_p$ to $[x]_N$. In particular, $([e_1]_N,\ldots,[e_n]_N)$ is an $\ff_p$-basis of $M/pM$ (this fact may also be proven directly), and the representative matrix of $\ov{b}$ is the reduction of $B$ modulo $p$.

We then get $\det(\ov{b})=[\det(B)]_p=[\det(M)]_p\in\ff_p$. In particular, $\ov{b}$ is nondegenerate if and only if $p\nmid \det(M).$
\end{ex}

Before continuing, let us note that the existence of $N$ has strong consequences on the determinant of $M$, as the next lemma shows.

\begin{lem}\label{lem-det}
Let $M$ be a full integral lattice of $E$. Assume that $N$ is a sublattice of $M$ such that $pM\subset N$ and $b(x,y)\in p\zz$ for all $x\in M$ and $y\in N$.
Then $N$ is a full lattice, and $p^{n-\dim_{\ff_p}(M/N)}$ divides $\det(M)$.

In particular, if $p\nmid \det(M)$, then the only sublattice $N$ satisfying the required conditions is $pM$, $\dim_{\ff_p}(M/N)=n$, and the corresponding $\ff_p$-bilinear form $\ov{b}$ is nondegenerate in this case.
\end{lem}

\begin{proof}
Since $M/N$ is a finite dimensional $\ff_p$-vector space, we have $$M/N\simeq \zn{p}\times\cdots\times \zn{p}$$ as abelian groups.
It follows that the similarity invariants of $N$ (relatively to $M$) are all equal to $1$ or $p$.
Thus, there exists a $\zz$-basis $(e_1,\ldots,e_n)$ of $M$ and an integer $r\geq 0$ such that $(e_1,\ldots,e_r, pe_{r+1},\ldots,p e_n)$ is a $\zz$-basis of $N$. In particular, $M/N\simeq \ff_p^{n-r}$ and $\dim_{\ff_p}(M/N)=n-r$.

Now let $B=(b(e_i,e_j))_{i,j}$. Let $j\in\left\llbracket 1,r\right\rrbracket$. Since $e_j\in N$, the assumption on $N$ implies that $b(e_i,e_j)\in p\zz$ for all $i\in\left\llbracket 1,n\right\rrbracket.$ It follows that the $r$ first columns of $B$ lie in $p\zz^n$. Therefore, $p^r\mid \det(M)$, which is exactly what we wanted to prove.

If $p\nmid \det(M)$, the previous point imposes $r=0$, so $(pe_1,\ldots,pe_n)$ is a $\zz$-basis of $N$, meaning that $N=pM$. 
The rest follows from the previous example.
\end{proof}

\subsection{Lattices from Codes}
If $C\subset M/N$ is a code on $M/N$ (that is, an $\ff_p$-linear subspace of $M/N$), we set $$\Gamma_C=\tfrac{1}{\sqrt{p}}\pi^{-1}(C).$$ Note that $\pi^{-1}(C)$ a full sublattice of $M$ containing $N$. Recall that if $C$ is a code on $M/N$, the codimension of $C$ is the integer ${\rm codim}_{\ff_p}(C)=\dim_{\ff_p}(M/N)-\dim_{\ff_p}(C).$ We define the dual code $C^\perp$ with respect to $\ov{b}$: $$C^\perp=\{[x]_N \in M/N~|~ \ov{b}([x]_N,[y]_N)=[0]_N \mbox{ for all }[y]_N \in C \}.$$ If $C\subset C^\perp$, the code is said to be self-orthogonal. If $C= C^\perp$, the code is said to be self-dual.
\begin{thm}\label{thm-gamma}
Let us keep the previous notation, and let $C,C_1,C_2$ be codes on $M/N$. Then:

\begin{enumerate}
\item $\det(\Gamma_C)=\dfrac{\det(M)}{p^{n-2\codim_{\ff_p}(C)}};$

\item $\Gamma_{C_1}\subset \Gamma_{C_2}$ if and only if $C_1\subset C_2$;

\item $\Gamma_{C_1}\subset \Gamma^\dl_{C_2}$ if and only if $C_1\subset C_2^\perp$;

\item $\Gamma_C$ is an integral lattice if and only if $C\subset C^\perp$.
In particular, $\Gamma_C$ is unimodular if and only if $C\subset C^\perp$ and $\det(M)=p^{n-2\codim_{\ff_p}(C)}$.

\item $\Gamma_{C^\perp}\subset \Gamma^\dl_{C}$, and we have
$$[\Gamma_C^\dl:\Gamma_{C^\perp}]=\dfrac{\det(M)}{p^{n-(\codim_{\ff_p}(C)+\codim_{\ff_p}(C^\perp))}}.$$
In particular, if $\ov{b}$ is nondegenerate, we have $$[\Gamma_C^\dl:\Gamma_{C^\perp}]=\dfrac{\det(M)}{p^{n-\dim_{\ff_p}(M/N)}}.$$

\end{enumerate}
\end{thm}

\begin{proof}
Let $C,C_1,C_2$ be codes on $M/N$.
The first isomorphism theorem shows that the canonical projection $\pi:M\to M/N$ induces an isomorphism of abelian groups $M/\pi^{-1}(C)\simeq (M/N)/C,$ and one may easily check that it is an isomorphism of $\ff_p$-vector spaces. In particular, we get $$[M:\pi^{-1}(C)]=[M/N: C]=\dfrac{\vert M/N\vert}{\vert C\vert}=\dfrac{p^{\dim_{\ff_p}(M/N)}}{p^{\dim_{\ff_p}(C)}}.$$
Now, we have $\det(\Gamma_C)=\det(\tfrac{1}{\sqrt{p}}\pi^{-1}(C))=\dfrac{1}{p^n}\det(\pi^{-1}(C))$, and therefore 
$$\det(\Gamma_C)=\dfrac{1}{p^n}[M:\pi^{-1}(C)]^2\det(M),$$ that is $$\det(\Gamma_C)=\dfrac{\det(M)}{p^{n-2(\dim_{\ff_p}(M/N)-\dim_{\ff_p}(C))}}.$$
Hence we get $(1)$.
Now we have $$\Gamma_{C_1}\subset \Gamma_{C_2}\iff \pi^{-1}(C_1)\subset \pi^{-1}(C_2)\iff C_1\subset C_2,$$ the last equivalence coming from the fact that $\pi$ is surjective, so we have $(2)$.

We also have $\Gamma_{C_1}\subset \Gamma^\dl_{C_2}$ if and only if for all $x\in \pi^{-1}(C_1)$, and for all $y\in \pi^{-1}(C_2)$, we have $b(\tfrac{1}{\sqrt{p}}x,\tfrac{1}{\sqrt{p}}y)\in\zz.$ This is equivalent to say that  for all $x\in \pi^{-1}(C_1)$, and for all $y\in \pi^{-1}(C_2)$, we have
$b(x,y)\in p\zz$, that is $\ov{b}([x]_N,[y]_N))=[0]_p\in\ff_p$ for all $x\in \pi^{-1}(C_1)$, and for all $y\in \pi^{-1}(C_2)$.
Since $\pi$ is surjective, this is equivalent to  $\ov{b}(c_1,c_2)=[0]_p$  for all $c_1\in C_1$ and all $c_2\in C_2$, that is $C_1\subset C_2^\perp$.
Thus, we have proved $(3)$. Items $(4)$ and $(5)$ are then direct consequences of $(3)$.
Using $(1)$, we get 

$$\begin{array}{lll}[\Gamma_C^\dl:\Gamma_{C^\perp}]^2&=&\dfrac{\det(\Gamma_{C^\perp})}{\det(\Gamma_C^\dl)}\vspace{5pt} \cr 
&=& \det(\Gamma_{C^\perp})\det(\Gamma_C)\vspace{5pt}  \cr

&=& \dfrac{\det(M)^2}{p^{2n-2(\codim_{\ff_p}(C)+\codim_{\ff_p}(C^\perp))}},
\end{array}$$
hence the desired equality. If $\ov{b}$ is nondegenerate, by \cite[Lemma 3.11., p.9]{Sch}, we have $$\dim_{\ff_p}(C)+\dim_{\ff_p}(C^\perp)=\dim_{\ff_p}(M/N),$$
and therefore, $$\codim_{\ff_p}(C)+\codim_{\ff_p}(C^\perp)=\dim_{\ff_p}(M/N).$$
This concludes the proof.
\end{proof}

\begin{coro}\label{coro-det1}
Assume that $\det(M)=1$. Then for all codes $C$ on $M/N$, we have $\Gamma_C^\dl=\Gamma_{C^\perp}$.

In particular, $\Gamma_C$ is unimodular if and only if $C^\perp=C$. 
\end{coro}

\begin{proof}
Let $C$ be a code on $M/N$.
If $\det(M)=1$, then $N=pM$, $\dim_{\ff_p}(M/N)=n$, and $\ov{b}$ is nondegenerate by Lemma \ref{lem-det}. The last point of Theorem \ref{thm-gamma} shows that $[\Gamma^\dl_C:\Gamma_{C^\perp}]=1,$ and thus $\Gamma^\dl_C=\Gamma_{C^\perp}$.

Therefore, we have $$\Gamma^\dl_C=\Gamma_C\iff \Gamma_{C^\perp}=\Gamma_C\iff C^\perp=C,$$ the last equivalence following from Theorem \ref{thm-gamma} $(2)$.
\end{proof}

The above theorem and its corollary generalize the following well-known results:

(1)\cite[Prop.~1.3]{Ebeling}
For $E=\rr^n$, $M=\zz^n$, $N=2\zz^n$ and $b$ the standard inner product, since $p=2 \nmid \det(M)=1$, $M/N$ has dimension $n$ by Lemma \ref{lem-det} and $\bar{b}$ is the the standard inner product modulo $2$ which is nondegenerate.  
Then $\det(\Gamma_C)=\tfrac{1}{2^{-n+2\dim_{\ff_p}(C)}}=2^{n-2\dim_{\ff_p}(C)}$, and $C\subset C^\perp$ if and only if $\Gamma_C$ is an integral lattice. Furthermore, $C=C^\perp$ if and only if $\Gamma_C$ is unimodular.

(2) \cite[Lemma~5.5,~Prop.~5.2]{Ebeling}
For $\zeta_p$ a primitive $p$th root of unity, $M=\zz[\zeta_p]^m$ and $N=(1-\zeta_p)\zz[\zeta_p]^m$, we note that there is a ring isomorphism $$u:\zz[\zeta_p]\overset{\sim}{\to}(1-\zeta_p)\zz[\zeta_p]\simeq \ff_p$$ sending the class of $\zeta$ to $[1]_p$, and thus the inertia degree is 1, implying that the  $\qq$-automorphisms of $\mathbb{Q}(\zeta_p)$ reduce to the identity modulo $p$. In particular, so does complex conjugation $^*$. Then consider the bilinear form
$b:M\times M\to\zz$ defined by $$b(x,y)=\sum_{j=1}^m\Tr_{\qq(\zeta_p)/\qq}(x^*_jy_j) \ \mbox{ for all }x=(x_1,\ldots,x_m), y=(y_1,\ldots,y_m)\in M.$$ Then, reducing modulo $p$, we get $$[b(x,y)]_p = (p-1)\ds\sum_{j=1}^m u(x_j)u(y_j)=-\sum_{j=1}^m u(x_j)u(y_j).$$
Hence, the induced bilinear map $\ov{b}:\ff_p^m\times\ff_p$ is canonically isomorphic to $-\langle \, , \rangle,$ where $\langle \, , \rangle$ is the standard unit form on $\ff_p^m$.

For $C$ a code in $M/N$, we then get $$\begin{array}{lll} C^\perp&=& \{x \in M/N~|~ \ov{b}([x]_N,[y]_N)=[0]_p \mbox{ for all }y\in C\}\cr &=&\{x \in M/N~|~\langle u(x),u(y) \rangle =[0]_p \mbox{ for all }y\in C\}.\end{array}$$ Hence, $C^\perp$ is canonically isomorphic to the  dual for the standard unit form via $u$. 

Then $\Gamma_{C^\perp}\subset \Gamma_C^\dl$, and $\det(\Gamma_C)=\tfrac{p^{m(p-2)}}{p^{m(p-1)-2m+2\dim_{\ff_p}(C)}}=p^{m-2\dim_{\ff_p}(C)}$. Also if $C\subset C^\perp$, then $\Gamma_C$ is integral. If $C$ is self-dual, then $\Gamma_C$ is unimodular. Similar results hold for some totally real and CM fields \cite{KOO}.

\subsection{Duality and Metabolic Forms} Let $(V,\varphi)$ be a symmetric nondegenerate bilinear form over a field $K$. A subspace $W\subset V$ is  isotropic if some nonzero vector $x\in W$ satisfies $\varphi(x,x)=0$. If $\varphi(x,x)=0$ for all vectors $x\in W$, then $W$ is totally isotropic. If no nonzero vector $x\in W$ satisfies $\varphi(x,x)=0$, then $W$ is  anisotropic.   
We say that $(V,\varphi)$ is {\it metabolic} if there exists a subspace $W$ of $V$ such that $W^\perp=\{x\in V~|~\varphi(x,y)=0\mbox{ for all }y\in W\}=W.$ When $K$ has characteristic different from $2$, this is equivalent to say that $(V,\varphi)$ is hyperbolic, that is isomorphic to an orthogonal sum of hyperbolic planes (a hyperbolic plane of $V$ is an isotropic subspace of $V$ of dimension 2). If $K$ has characteristic $2$, there exist metabolic forms which are not hyperbolic.

In our context, we have $V=M/N$, $K=\ff_p$ and $\varphi = \bar{b}$.
The condition $C\subset C^\perp$ means that $C$ is a totally isotropic subspace of $M/N$ with respect to $\ov{b}$, while the equality $C^\perp=C$ is equivalent to say that $\ov{b}$ is a metabolic form. The structure of symmetric bilinear forms over $\ff_p$ is well understood, and we summarize next the known results.

Let $\varphi:V\times V\to\ff_p$ be a nondegenerate symmetric bilinear form.

\begin{enumerate}
\item
If $p$ is odd, then the anisotropic part of $\varphi$ has dimension $\leq 2$ (\cite[Theorem 3.3., p. 38]{Sch}).
In particular, if $\dim(\varphi)=2m-1$ or $2m,$ with $m\geq 2$,  then $\varphi$ splits off a hyperbolic (hence metabolic) summand of dimension $2(m-1)$. Therefore, for all $d\in\left\llbracket 1,m-1\right\rrbracket$, there exists at least one subspace $W$ of dimension $d$ of $V$ satisfying $W\subset W^\perp$.
\item
It also follows that there exists at least one subspace $W$ of $V$ such that $W^\perp =W$ if and only if $\dim_{\ff_p}(\varphi)=2m$ and $\det(\varphi)$ lies in the square class of $[(-1)^m]_p$. The dimension of $W$ is $m$.
\item
Assume now that $p=2$.  Then the anisotropic part of $\varphi$ has dimension $\leq 1$ (\cite[Theorem 1.6., p. 170]{Sch}).
Therefore, as above, for all $d\in\left\llbracket 1,m-1\right\rrbracket$, there exists at least one subspace $W$ of dimension $d$ of $V$ satisfying $W\subset W^\perp$.

Moreover, there exists at least one subspace $W$ of $V$ such that $W^\perp =W$ if and only if $\dim_{\ff_p}(\varphi)=2m$. The dimension of $W$ is $m$. 
\end{enumerate}

These results, together with Theorem \ref{thm-gamma} and Corollary \ref{coro-det1} show that it is not difficult to construct integral lattices or unimodular lattices: given a lattice $M$ of $(E,b)$ and $p$ a prime, choose $N=pM$, compute the $\ff_p$-vector space $M/N$, find the appropriate subspace of $M/N$ according to $\bar{b}$ and lift it via $\pi$. While we have $\Gamma_{C_1}=\Gamma_{C_2}$ if and only if $C_1=C_2$ by Theorem \ref{thm-gamma}, note however that this is not true at the level of isometries. More precisely, the fact that $(C_1, \ov{b}_{C_1\times C_1})$ and $(C_2, \ov{b}_{C_2\times C_2})$ are isomorphic bilinear spaces does not necessarily imply that the lattices $\Gamma_{C_1}$ and $\Gamma_{C_2}$ are isomorphic.

We give two examples.

\begin{ex}
We endow $E=\rr^2$ with its standard inner product. Let $M=\zz^2,$ $p=7$ and $N=7\zz^2$. Then $M/N$ canonically identifies to $\ff_7^2$, $\pi:\zz^2\to\ff_7^2$ is the canonical projection and the induced $\ff_7$-bilinear form on $M/N$ is the standard bilinear form on $\ff_7^2$. Consider the codes $C_1,C_2$ given by
$$C_1=\ff_7\cdot \begin{pmatrix}[1]_7 \cr [0]_7\end{pmatrix},~C_2=\ff_7\cdot \begin{pmatrix}[1]_7 \cr [1]_7\end{pmatrix}.
$$
Then  $C_1$ and  $C_2$ are isomorphic bilinear spaces. Indeed the corresponding bilinear forms are respectively the diagonal forms $\langle [1]_7\rangle$ and $\langle [2]_7\rangle$, which are isomorphic since $[2]_7$ is a square in $\ff_7^\times$. 

The respective $\zz$-basis of $\pi^{-1}(C_1)$ and $\pi^{-1}(C_2)$ are
$$
(\begin{pmatrix}
1 \cr 0\end{pmatrix},  \begin{pmatrix}
0 \cr 7\end{pmatrix}),~(\begin{pmatrix}
1 \cr 1\end{pmatrix},  \begin{pmatrix}
0 \cr 7\end{pmatrix}).$$

The matrices of these lattices in the corresponding bases are 
$$\begin{pmatrix}
1 & 0 \cr 0& 49
\end{pmatrix},~\begin{pmatrix}
2 & 7 \cr 7& 49
\end{pmatrix}.$$ 
Clearly, the first integral bilinear form cannot represent $2$ over $\zz^2$, while the second one does. Hence, these two lattices are not isomorphic, so $\Gamma_{C_1}$ and $\Gamma_{C_2}$ are not isomorphic either. 
\end{ex}

\begin{ex}
We endow $E=\rr^8$ with its standard  inner product.  Let $M=\zz^8,$ $p=2$ and $N=2\zz^8$. Let $(e_1,\ldots, e_8)$ be the canonical basis of $\ff_2^8$. Let $C_1$ be the span of the four vectors $$e_1+e_2, e_3+e_4, e_5+e_6,e_7+e_8.$$ 
These vectors are isotropic and orthogonal to each other. They thus form a totally isotropic subspace $C_1$ of dimension 4, and $C_1 = C_1^\perp$.
Let $C_2$ be the span of the four vectors $$e_1+e_2+e_3+e_4, e_1+e_2+e_5+e_6, e_1+e_2+e_7+e_8, e_1+e_3+e_5+e_7.$$
These vectors are isotropic. Since every vector is a sum of four $e_i$'s, and each pair of vectors has exactly two $e_i$'s in common, these vectors are orthogonal to each other, and they also form a totally isotropic subspace $C_2$ of dimension 4, with $C_2 = C_2^\perp$. 
Both induced bilinear forms are zero, hence they are isomorphic.
Corollary \ref{coro-det1} tells us that the corresponding lattices $\Gamma_{C_1}$ and  $\Gamma_{C_2}$ are unimodular lattices. However, they are not isomorphic. 
Indeed $\Gamma_{C_1}$ is an odd unimodular lattice (it contains for example the vector $x=(1,1,0,0,0,0,0,0)$ such that $b(x,x)=\tfrac{1}{2}\langle x,x\rangle=1$ which is not even), while $\Gamma_{C_2}$ is an even unimodular lattice (we have $b(x,x)$ even for all $x\in\Gamma_{C_2}$).
Note that $\Gamma_{C_2}$ is an even unimodular lattice of rank $8$, so it is isomorphic to $E_8$.
\end{ex}

The main difference with respect to the existing literature is that the dual of a code $C$ is usually studied with respect to the standard inner product, sometimes with respect to a Hermitian standard inner product if the code alphabet considered possesses a suitable automorphism. Therefore the usual philosophy is to start with the standard inner product and then to look for a self-dual code, while the discussion above shows the existence of a self-dual code with respect to some bilinear form $\bar{b}$, which may or not be the standard inner product for the code, depending on the choice of $b$.

\subsection{Considered Cases}

Let $\mathcal{A}$ be a finite dimensional (not necessarily commutative) associative unital $\qq$-algebra, equipped with a positive definite symmetric $\qq$-bilinear form $b:\mathcal{A}\times\mathcal{A}\to\qq.$ Set $E=\mathcal{A}\otimes_\qq\rr$. Then, $E$ is a finite dimensional real vector space and $b_\rr$ is a positive definite $\rr$-bilinear form on $E$.

Now, let $\Lambda$ be a $\qq$-order of $\mathcal{A}$, that is a subring of $\mathcal{A}$ with is also a free abelian group generated by a $\qq$-basis of $\mathcal{A}$. This is equivalent to say that $\Lambda$ is a subring of $\mathcal{A}$ such that we have a $\qq$-algebra isomorphism $\Lambda\otimes_\zz \qq\simeq \mathcal{A}$.
We then have $$\Lambda\otimes_\zz\rr\simeq (\Lambda\otimes_\zz \qq)\otimes_\qq \rr\simeq \mathcal{A}\otimes_\qq\rr=E.$$

Then $M=\Lambda\otimes_\zz 1$ is a full sublattice of $E$, and it is easy to check that its dual with respect to $b_\rr$ is $\Lambda^\dl\otimes_\zz 1$, where $$\Lambda^\dl=\{ x\in\mathcal{A}\mid b(x,y)\in\zz \mbox{ for all }y\in\Lambda\}.$$

In the sequel, we will only consider lattices $M$ of the form $M=\Lambda\otimes_\zz 1$.

Hence, any left/right ideal $\mathcal{J}$ of $\Lambda$ canonically identifies to a full sublattice $M$ of $E$ (namely $M=\mathcal{J}\otimes_\zz 1$). Modulo this canonical identification, the restriction of the inner product to $\mathcal{J}\times \mathcal{J}$ is exactly the restriction of $b$ to $\mathcal{J}\times \mathcal{J}$, meaning that we may ignore $E$ and work directly with $b$ for computations (rather than $b_\rr$)

The setting $pM \subset N \subset M$ of this section is then applied to a two-sided ideal $\mathcal{I}$ and a left ideal $\mathcal{J}$ of $\Lambda$
such that $p\mathcal{J}\subset \mathcal{I}\subset \mathcal{J}$,
 and we still assume that $b\in \zz$ for all $x,y\in \mathcal{J}$ and 
$b(x,y)\in p\zz$ for all $x\in \mathcal{J}$ and all $y\in\mathcal{I}$.

Then the \emph{dual} $\mathcal{J}^\dl$ of $\mathcal{J}$ becomes
$$\mathcal{J}^\dl=\{x\in \mathcal{A}\mid b(x,y)\in \zz \mbox{ for all }y\in \mathcal{J}\}.$$

In this context, we will make the following abuse of notation : strictly speaking, if $C$ is a code of $\mathcal{J}/\mathcal{I}$, $\Gamma_C$ is $\pi^{-1}(C)\otimes \dfrac{1}{\sqrt{p}}$. However, we will still denote it by $\dfrac{1}{\sqrt{p}}\pi^{-1}(C)$, even in a noncommutative setting. This is not a serious matter, since at the level of bilinear form, we just divide by $p$.

Our main motivation for choosing $\mathcal{J}$ a left ideal of $\Lambda$ is that for any left ideal $C$ of $\Lambda/\mathcal{I}$ contained in $M/N$, $\pi^{-1}(C)$ 
will be a left ideal of $\Lambda$ contained in $\mathcal{J}$. In particular, the lattice $\Gamma_C$ will inherit an extra multiplicative structure, in the sense that for any $a\in\Lambda$ and any $x\in\Gamma_C$, then $ax\in\Gamma_C$.

\begin{rem}  
Notice that because of the normalisation factor $\dfrac{1}{\sqrt{p}}$, the product of two elements of $\Gamma_C$ will not be necessarily an element of $\Gamma_C$ (nevertheless, this will be true for elements of $\pi^{-1}(C)$).

\end{rem}

Most of the time, we will be in the easiest situation where $\mathcal{J}=\Lambda$.

The two cases that we will study in particular are:
\begin{itemize}
\item
$\mathcal{A}=L$ is a Galois totally real or CM number field, $\Lambda=\OO_L$ is the ring of integers of $L$, and $q_{L,\lambda}(x,y)=\Tr_{L/\qq}(\lambda x^*y)$,where $\lambda$ is a suitable real parameter;

\item
$\mathcal{A}=B=(\gamma, L/k,\sigma)$ is a cyclic $k$-algebra, with $L/\qq$  Galois totally real or CM, and $q_{B,\lambda}(x,y)=\Tr_{L/\qq}(\Trd_B(\lambda\tau(x)y))$, where  $\lambda$ is a suitable real parameter and
$\tau$ is an involution on $B$.
\end{itemize}

%
%
%
\section{Polynomial Codes and Lattices over Number Fields}

We now assume that $L$ is a number field, such that complex conjugation $^*$ induces an automorphism of $L$.  We fix a prime number $p$, an ideal $I$ of $\mathcal{O}_L$ containing $p$ such that $I^*=I$, and a $\zz$-linear map $s:\mathcal{O}_L\to\zz$. 

Notice that $\mathcal{O}_L/I$ has a natural structure of an $\ff_p$-algebra, since $p\in I$, and thus so have ideals of $\mathcal{O}_L/I$.  We will make several assumptions:

$(H_1)$ The linear map $s$ induces on $\mathcal{O}_L/I$ a well-defined nondegenerate symmetric $\ff_p$-bilinear map $$\nfuncp{\varphi}{\mathcal{O}_L/I\times\mathcal{O}_L/I}{\ff_p}{([x]_I,[y]_I)}{[s(x^*y)]_p}$$

$(H_2)$ There exists a nonzero monic polynomial $\ov{\mu}\in\ff_p[X]$ such that we have an isomorphism of $\ff_p$-algebras $$\ff_p[X]/(\ov{\mu})\simeq \mathcal{O}_L/I.$$

We will see concrete examples where these assumptions are satisfied later on. For now, we are going to investigate the ideals of  $\mathcal{O}_L/I$.

\subsection{Ideals of $\mathcal{O}_L/I$ and Codes}\label{ssec-codes}
The isomorphism of $\ff_p$-algebras $\ff_p[X]/(\ov {\mu})\simeq \mathcal{O}_L/I$ shows that ideals of $\mathcal{O}_L/I$ correspond to ideals of $\ff_p[X]$ containing $\ov{\mu}$, which themselves correspond to monic divisors of $\ov{\mu}$, which in turn may be used as generator polynomials of polynomial codes.

Notice that the assumption boils down to have a surjective ring morphism $\ff_p[X]\to\mathcal{O}_L/I$. Since $[1]_p$ is necessarily mapped onto the unit element of $\mathcal{O}_L/I$, this morphism is a morphism of $\ff_p$-algebras, and is necessarily given by evaluation at a class $[\beta]_I\in\mathcal{O}_L/I$. The kernel of this surjective morphism is generated by $\ov{\mu}$.

Let $\beta\in\mathcal{O}_L$ be such that the class of $X$ is mapped onto the class of $\beta$. Then the isomorphism above is given by $$\funcisop{\ff_p[X]/(\ov{\mu})}{\mathcal{O}_L/I}{[\ov{f}]_{\ov{\mu}}}{[f(\beta)]_I}$$
Note that this isomorphism is well-defined because $p\in I$, so the result does not depend on the choice of the representative $f$.

 If $\ov{g}$ is a monic divisor of $\ov{\mu}$, the previous considerations show that it corresponds to the ideal $\OO_L/I\cdot[g(\beta)]_I$, and thus to the ideal $\mathcal{O}_L \, g(\beta)+I$ of $\mathcal{O}_L.$ 

Since complex conjugation is an automorphism of $L$, it induces an automorphism of $\mathcal{O}_L$, and in turn an automorphism of  $\mathcal{O}_L/I$, since $I^*=I,$ still denoted by $^*$. By definition, we have $[x]_I^*=[x^*]_I$ for all $x\in\mathcal{O}_L$. Therefore complex conjugation also induces  a correspondence between ideals, hence between monic divisors of $\ov{\mu}$.  If $\ov{g}$ is such a monic divisor, we will denote by $\ov{g}_*$ the corresponding monic divisor of $\ov{\mu}$.

Notice that if the ideal $I'/I$ corresponds to the ideal generated by $\ov{g}$, then the $\ff_p$-vector space $I'/I$  has dimension $\dim_{\ff_p}(I'/I)=\deg(\ov{\mu})-\deg(\ov{g})$, and so has the code with generator polynomial $\ov{g}$.  
Since $^*$ is an automorphism of $\mathcal{O}_L/I$, $(I'/I)^*$ and $(I'/I)$ have same dimension over $\ff_p$. In particular, $\ov{g}_*$ and $\ov{g}$ have same degree. 

\begin{thm}\label{dualpoly}
Assuming $(H_1)$ and $(H_2)$, let $\ov{g}$ be a monic divisor of $\ov{\mu}$. Let $I'/I$ be the corresponding ideal of $\mathcal{O}_L/I.$  Then $$(I'/I)^\perp=\{\ov{y} \in \mathcal{O}_L/I~|~\varphi([x]_I,[y]_I)=0 \mbox{ for  all }[x]_I \in I'/I\}$$ is an ideal of $\mathcal{O}_L/I$, corresponding to the monic divisor $\ov{g}_\perp=\dfrac{\ov{\mu}}{\ov{g}_*}$. 

In particular, $I'/I$ is self-orthogonal if and only if $\ov{\mu}\mid \ov{g}_*\ov{g}$, and self-dual if and only if $\ov{g}_*\ov{g}=\ov{\mu}$.
\end{thm}

\begin{proof}
By some previous considerations, we have $I'=\mathcal{O}_L\, g(\beta)+I$. Let $[y_1]_I, [y_2]_I\in (I'/I)^\perp$, and let $[a]_I\in\mathcal{O}_L/I$. For all $[x]_I\in I'/I$, we have $$\begin{array}{lll}\varphi([x]_I,[y_1]_I+ [a]_I[y_2]_I)&=&[s(x^*y_1)+s(ax^*y_2)]_p\cr 
&=& \varphi([a]_I^* [x]_I,[y_2]_I)+\varphi([x]_I,[y_1]_I)\cr 
&=&[0]_p,\end{array}$$
since $I'/I$ is an ideal and  $[y_1]_I, [y_2]_I\in (I'/I)^\perp$. Thus,  $[y_1]_I+ [a]_I[y_2]_I\in (I'/I)^\perp$ and $(I'/I)^\perp$ is an ideal.

Hence $(I'/I)^\perp=I''/I$, where $I''=\mathcal{O}_L h(\beta)+I$, for some monic divisor $\ov{h}$ of $\ov{\mu}$.
Hence any element of  $(I'/I)^\perp$ has the form $[y h(\beta)]_I,$ where $y\in\mathcal{O}_L$.

Now, for all $x\in I'$, we have $$0=\varphi([x]_I,[y h(\beta)]_I)=[s(x^* y h(\beta))]_p=[s(h(\beta) x^* y]_p=\varphi([x]_I[h(\beta)]_I^*,[y]_I)$$ for all $[y]_I\in \mathcal{O}_L/I.$
Since $\varphi$ is nondegenerate, we get that $$[x]_I[h(\beta)]_I^*=[0]_I\in\mathcal{O}_L/I  \mbox{ for all }[x]_I\in I'/I,$$ that is $$[x]_I^* [h(\beta)]_I=[0]_I\in\mathcal{O}_L/I \mbox{ for all }[x]_I\in I'/I.$$ This is equivalent to say that $[z]_I [h(\beta)]_I=[0]_I\in\mathcal{O}_L/I$  for all $[z]_I \in (I'/I)^*$.
By definition of $g_*$,  $(I'/I)^*$ is generated by $[g_*(\beta)]_I$. In particular, we have $[g_*(\beta)]_I [h(\beta)]_I=[0]_I\in\mathcal{O}_L/I.$ Using the $\ff_p$-algebra isomorphism $\ff_p[X]/(\ov{\mu})\simeq \mathcal{O}_L/I,$ we see that it is equivalent to $\ov{g}_* \ov{h}\equiv 0 \mod (\ov{\mu}),$ that is $\ov{\mu}\mid \ov{g}_*\ov{h}$. Thus $\ov{g}_\perp\mid\ov{h}$.

Since $\varphi$ is nondegenerate, we have $$\dim_{\ff_p}((I'/I)^\perp)=\dim_{\ff_p}(\mathcal{O}_L/I)-\dim_{\ff_p}(I'/I), $$and thus $$\deg(\ov{h})=\deg(\ov{\mu})-\deg(\ov{g})=\deg(\ov{\mu})-\deg(\ov{g}_*)=\deg(\ov{g}_\perp).$$
Since $\ov{g}_\perp\mid\ov{h}$ and these polynomials are monic, we get $\ov{h}=\ov{g}_\perp$.

Now the inclusion $(I'/I)\subset (I'/I)^\perp$ corresponds to the inclusion $(\ov{g})\subset (\ov{g}_\perp)$, that is $\ov{g}_\perp\mid \ov{g}$, while the equality $(I'/I)= (I'/I)^\perp$ corresponds to the equality $\ov{g}=\ov{g}_\perp$ (since the polynomials are monic). The last part of the theorem follows immediately.
\end{proof}

We now give examples of situations where $(H_2)$ holds.

\begin{prop}\label{proph20}
Let $I$ be an ideal of $\mathcal{O}_L$ such that the $\ff_p$-algebra $\mathcal{O}_L/I$ is generated by the class of an element $\alpha\in \mathcal{O}_L$.
Then there exists an isomorphism of $\ff_p$-algebras $$\ff_p[X]/(\ov{\mu})\simeq \mathcal{O}_L/I,$$
which sends the class of $X$ to the class of  $\alpha$. 

Moreover, the polynomial $\ov{\mu}$ is the unique monic polynomial of smallest degree such that $\mu(\alpha)\in I$.
\end{prop}

\begin{proof}
Since $I$ contains $p$, evaluation at $\alpha$ induces a morphism of $\mathbb{F}_p$-algebras $\theta:\mathbb{F}_p[X]\to \mathcal{O}_L/I$ which is surjective by assumption. Moreover, $\ker(\theta)$ is the set of polynomials $\ov{f}$ such that $f(\alpha)\in I$. It is generated by the unique monic polynomial of $\ker(\theta)$ of smallest degree. Now apply the first isomorphism theorem to conclude.
\end{proof}

We would like  to compute explicitly $\ov{g}_*$ in various cases.

\begin{lem}
Assume that complex conjugation induces an automorphism on $L$. 
Let $I$ be an ideal of $\mathcal{O}_L$ such that $I^*=I.$
Assume that the $\ff_p$-algebra $\mathcal{O}_L/I$ is generated by the class of an element $\alpha\in \mathcal{O}_L$.

Let  $\ov{\mu}$ be the unique monic polynomial of smallest degree such that $\mu(\alpha)\in I$, so that we have $$\ff_p[X]/(\ov{\mu})\simeq \mathcal{O}_L/I,$$
where the class of $X$ is mapped to the class of $\alpha$.
 
Let $\ov{g}\in\ff_p[X]$ be a monic divisor of $\ov{\mu}$. There exists a polynomial $g_0\in\zz[X]$ such that $g(\alpha)^*-g_0(\alpha)\in I$ . Then $\ov{g}_*=gcd(\ov{g}_0,\ov{\mu})$.
\end{lem}

\begin{proof}
We have $[g(\alpha)^*]_I\in \OO_L/I$, so by assumption on $\OO_L/I$, there exists $\ov{g}_0\in\ff_p[X]$ such that $[g(\alpha)^*]_I=\ov{g}_0([\alpha]_I)= [g_0(\alpha)]_I$. Then $g^*(\alpha)-g_0(\alpha)\in I$.

The ideal of $\OO_L/I$ corresponding to $\ov{g}$ being $\OO_L/I\cdot [g(\alpha)]_I$, the ideal corresponding to $\ov{g}_*$ is by definition 
$$
(\OO_L/I\cdot [g(\alpha)_I])^*=\OO_L/I\cdot [g(\alpha)^*]_I=\OO_L/I\cdot [g_0(\alpha)]_I.$$
But, this ideal also corresponds to the ideal $((\ov{g}_0)+(\ov{\mu}))/(\ov{\mu})$ of $\ff_p[X]/(\ov{\mu})$, that is $gcd(\ov{g}_0,\ov{\mu})/(\ov{\mu})$. This yields the desired result.
\end{proof}

\begin{prop}\label{calculgstar}
Let $L$ be a number field of degree $n$.
Assume that complex conjugation induces an automorphism on $L$.

Let $I$ be an ideal of $\mathcal{O}_L$ such that $I^*=I.$
Assume that the $\ff_p$-algebra $\mathcal{O}_L/I$ is generated by the class of an element $\alpha\in \mathcal{O}_L$.

Let  $\ov{\mu}$ be the unique monic polynomial of smallest degree such that $\mu(\alpha)\in I$, so that we have $$\ff_p[X]/(\ov{\mu})\simeq \mathcal{O}_L/I,$$
where the class of $X$ is mapped to the class of $\alpha$.
 
Let $\ov{g}\in\ff_p[X]$ be a monic divisor of $\ov{\mu}$.

\begin{enumerate}

\item Assume that $\alpha^*=\alpha.$ Then $\ov{g}_*=\ov{g}$.

\medskip

\item  Assume that $\alpha^*=-\alpha.$ Then $$\ov{g}_*=(-1)^{\deg(\ov{g})}\ov{g}(-X).$$

\medskip 

\item  Assume that $L=\qq(\sqrt{-d}),$ where $d$ is a positive squarefree integer and  $-d\equiv 1 \ [4]$. Then $$\ov{g}_*=(-1)^{\deg(\ov{g})}\ov{g}(1-X).$$

\medskip

\item  Assume that $\alpha^*\alpha=1$ Then $$\ov{g}_*=\ov{g}(0)^{-1}X^{\deg(\ov{g})}\ov{g}(X^{-1}).$$

\end{enumerate}
\end{prop}

\begin{proof}~

$(1)$ Since $\ov{g}_*=gcd(\ov{g}_0,\ov{\mu})$, where $g_0$ is such that $g(\alpha)^*-g_0(\alpha)\in I$, one can take $g_0=g$.

$(2)$ One can take $g_0=g(-X)$.

Notice that, since $\mu$ has integral coefficients, $\mu(\alpha^*)=\mu(\alpha)^*\in I^*=I$. Hence, $\mu(-\alpha)\in I.$ If we set $h=(-1)^n\ov{\mu}(-X)$, then $\ov{h}\in\ff_p[X]$ is a monic polynomial such that $h(\alpha)\in I$, of degree $\deg(\ov{\mu})$.
Hence $\ov{h}=\ov{\mu}$, that is $$(-1)^n\ov{\mu}(-X)=\ov{\mu}.$$
Since $\ov{g}$ divides $\ov{\mu}$, reducing modulo $p$, it follows from this equality that $\ov{g}_0$ divides $\ov{\mu}.$
Consequently, the monic gcd of  $\ov{g}_0$ and $\ov{\mu}$ is $(-1)^{\deg(\ov{g})}\ov{g}(-X).$
In other words, we get $$\ov{g}_*=(-1)^{\deg(\ov{g})}\ov{g}(-X).$$

$(3)$ Here $\alpha=\dfrac{1+\sqrt{-d}}{2}$ and $\alpha^*=1-\alpha$. 
Reasoning as before, and noticing that $\ov{g}$ and $\ov{g}(-1-X)$ have same degree,  we get the desired result.

$(4)$ We first show that $\ov{\mu}(0)\neq 0$. Otherwise, we would have $\ov{\mu}=X\ov{r}$ for some monic polynomial $\ov{r}\in\ff_p[X]$. 
Then $\alpha r(\alpha)\in I$, and since $I$ is an ideal, we get $\alpha^*\alpha r(\alpha)=r(\alpha)\in I$. This would contradict the minimality of $\deg(\ov{\mu})$.

In particular, the constant term of $\ov{g}$ is not zero.

Now, one can take $g_0=\ov{g}(0)^{-1}X^d g(X^{-1})$, where $d=\deg(g)=\deg(\ov{g})$. As in $(2)$, $$\mu(\alpha^*)=\mu(\alpha^{-1})\in I.$$ Let $[b]_p=[a_0]_p^{-1}$, where $[a_0]_p$ is the constant term of $\ov{\mu}$.
Set $h=bX^n\mu(X^{-1})$. Then $h(\alpha)=b\alpha^n \mu(\alpha^{-1})\in I$. But $\ov{h}$ is monic (since its leading coefficient is $[b]_p[a_0]_p=[1]_p$) of degree $\deg(\ov{\mu})$, and thus  $$[b]_pX^n\ov{\mu}(X^{-1})=\ov{\mu}.$$

Now write $\ov{g} \,\ov{r}=\ov{\mu}$.  If $\deg(\ov{h})=e$, then $d+e=n$, and we have 
$$X^n\ov{\mu}(X^{-1})=(X^d\ov{g}(X^{-1}))(X^e\ov{r}(X^{-1})).$$
 
The equality above shows that $\ov{g}_0$ is a divisor of $\ov{\mu}$. As before, we conclude that $$\ov{g}_*=\ov{g}(0)^{-1}X^{\deg(\ov{g})}\ov{g}(X^{-1}).$$
\end{proof}

\begin{rem}
All these results apply in particular when $\mathcal{O}_L=\zz[\alpha]$ for some $\alpha\in\mathcal{O}_L$.
\end{rem}

The next result gives a concrete example where the assumptions of Proposition \ref{proph20} are satisfied.

\begin{tikzcd}[column sep=small]
 & L = K_1K_2 \arrow[ld,dash] \arrow[rd,dash] & \\
K_1\supset \mathcal{O}_{K_1}=\mathbb{Z}[\alpha_1] \arrow[rd,dash,"n_1"'] &  & K_2 \arrow[ld,dash,"n_2"] \supset \mathcal{O}_{K_2}=\mathbb{Z}[\alpha_2], &\hspace{-0.8cm} p\mathcal{O}_{K_2} = \mathfrak{p}_2^{n_2}\\
 & \mathbb{Q} \supset \mathbb{Z} \ni p & \\
\end{tikzcd}

\begin{prop}\label{proph2}
Let $K_1, K_2$ be two numbers fields of degree $n_1$ and $n_2$ respectively which are arithmetically disjoint (that is  linearly disjoint over $\qq$ with coprime discriminants). Assume that $\mathcal{O}_{K_i}=\zz[\alpha_i]$ for $i=1,2.$
Let $p$ be a prime number which is totally ramified in $K_2$, and let $\pp_2$ be the unique prime ideal of $\mathcal{O}_{K_2}$ lying above $p$. Set $L=K_1K_2$.

Then  $\mathcal{O}_L/\pp_2\mathcal{O}_L$ is generated by the class of $\alpha_1$. 
Moreover, given the minimal polynomial $\mu_{\alpha_1,\qq}$ of $\alpha_1$, we have an isomorphism of $\ff_p$-algebras  $$\ff_p[X]/(\ov{\mu}_{\alpha_1,\qq})\simeq \mathcal{O}_L/\pp_2\mathcal{O}_L,$$
which maps the class of $X$ modulo $(\ov{\mu}_{\alpha_1,\qq})$
to the class of $\alpha_1$ modulo $\pp_2\mathcal{O}_L$,
and an isomorphism of $\ff_p$-algebras $$\ff_p[X]/(\ov{\mu}_{\alpha_1,\qq})\simeq \mathcal{O}_{K_1}/p\mathcal{O}_{K_1},$$
which maps the class of $X$ modulo $(\ov{\mu}_{\alpha_1,\qq})$ to the class of $\alpha_1$  modulo $p\mathcal{O}_{K_1}.$

In particular, we get have an isomorphism of $\ff_p$-algebras $$\mathcal{O}_{K_1}/p\mathcal{O}_{K_1}\simeq \mathcal{O}_L/\pp_2\mathcal{O}_L,$$
which sends the class of $x_1\in\mathcal{O}_{K_1}$ modulo $p\mathcal{O}_{K_1}$ to the class of $x_1$ modulo $\pp_2\mathcal{O}_L$.
\end{prop}

\begin{proof}
Since $\mathcal{O}_2=\zz[\alpha_2]$, a theorem of Dedekind asserts that reduction modulo $p$ of $\mu_{\alpha_2,\qq}$ is $(X-\ov{a}_2)^{n_2}\in\ff_p[X]$ for some $\ov{a}_2\in\ff_p$,
and that $p\mathcal{O}_{K_2}=\pp_2^{n_2}$, where $\pp_2=(\alpha_2-a_2,p)$.


Let us denote by $f:\mathcal{O}_{K_1}\to \mathcal{O}_L/\pp_2\mathcal{O}_{L}$ the ring morphism sending $x_1\in\mathcal{O}_{K_1}$ to its class modulo $\pp_2\mathcal{O}_L.$

The assumption on $K_1$ and $K_2$ also implies that the various powers $\alpha_1^i\alpha_2^j$ form a $\zz$-basis of $\mathcal{O}_L$.   Since $\alpha_2-a_2\in\pp_2\subset\pp_2\mathcal{O}_L $, we have $\alpha_2\equiv a_2\mod \pp_2\mathcal{O}_L$.
It readily follows that any element of $\mathcal{O}_L$ is congruent to an element of the form $P(\alpha_1),P\in\zz[X]$, which is an element of $\mathcal{O}_{K_1}$. In other words, the $\ff_p$-algebra is generated by the class of $\alpha_1$. By Proposition \ref{proph20}, evaluation at $\alpha_1$ induces an isomorphism of $\ff_p$-algebras  $$\ff_p[X]/(\ov{\mu})\simeq \mathcal{O}_L/\pp_2\mathcal{O}_L,$$ where $\ov{\mu}$ is the unique monic polynomial 
of smallest degree such that $\mu(\alpha_1)\in \pp_2\mathcal{O}_L$. 

Now we have $$\vert  \mathcal{O}_L/\pp_2\mathcal{O}_L\vert=N_{L/\qq}(\pp_2\mathcal{O}_L)=N_{K_2/\qq}(\pp_2)^{n_1}=p^{n_1}$$ since $p$ totally ramifies in $K_2$, so $\deg(\ov{\mu})=n_1$. But $\mu_{\alpha_1,\qq}$ has degree $n_1$, and $\mu_{\alpha_1,\qq}(\alpha_1)=0\in\pp_2\mathcal{O}_L$. Since $\ov{\mu}_{\alpha_1,\qq}$ is monic, we deduce that $\ov{\mu}=\ov{\mu}_{\alpha_1,\qq}$.

For the second isomorphism, apply Proposition \ref{proph20} to $L=K_1$ and $I=p\mathcal{O}_L$.
The rest of the proposition is then clear.
\end{proof}

We now prove that hypothesis $(H_1)$ is fulfilled in the context of Proposition \ref{proph2} when $p$ is tamely ramified.

\subsection{Ideals of Number Fields and Lattices}

We now give a concrete example of map $s:\OO_L\to\zz$ which will be used in the rest of the paper.
We first need a lemma.

\begin{lem}\label{traceprops}
Let $L$ be a number field. Then, for any $\lambda\in L^\times$, the symmetric $\qq$-bilinear map $$\nfunc{\mathcal{T}_{L,\lambda}}{L\times L}{\qq}{(x,y)}{\Tr_{L/\qq}(\lambda xy)}$$ is nondegenerate.

Moreover, $\mathcal{T}_{L,\lambda}$ is positive definite if and only if $L$ is totally real and $\lambda$ is totally positive. 
\end{lem}

\begin{proof}
Notice that $\Tr_{L/\qq}(1)=[L:\qq]\neq 0$. Now, let $x\in L$ such that $\Tr_{L/\qq}(\lambda xy)=0$ for all $y\in L$.
If $x\neq 0$, taking $y=(\lambda x)^{-1}$ yields a contradiction. Hence $x=0$ and $\mathcal{T}_{L,\lambda}$ is nondegenerate.
By \cite[3.2.8]{Sch}, the distinct orderings of $L$ are $$\sigma_1^{-1}(\rr^2),\ldots, \sigma_r^{-1}(\rr^2),$$
where $\rr^2=\{\lambda^2\mid \lambda\in\rr\}$ and $\sigma_1,\ldots,\sigma_r$ are the real $\qq$-embeddings of $L$. 
By \cite[3.4.5]{Sch}, we have $${\rm sign}(\mathcal{T}_{L,\lambda})=\vert \{i\in\left\llbracket 1,r\right\rrbracket \mid \sigma_i(\lambda)>0\}\vert
- \vert \{i\in\left\llbracket 1,r\right\rrbracket \mid \sigma_i(\lambda)<0\}\vert.$$

Then ${\rm sign}(\mathcal{T}_{L,\lambda})=[L:\qq]$ if and only if $r=[L:\qq]$ and $\sigma_i(\lambda)>0$ for $i\in\left\llbracket 1,r\right\rrbracket$, which is the desired result.
\end{proof}

Before stating the next result, we recall from \cite[\S 2]{Lang} the following two equivalent characterizations of a CM field: (1) $L$ is a totally imaginary quadratic extension of a totally real field, (2) complex conjugation commutes with every embedding of $L$ into an algebraic closure of $\qq$ as embedded into the complex numbers, and $L$ is not real.

\begin{lem}\label{trstarprops}
Let $L$ be a number field such that the complex conjugation $\ov{\phantom{\lambda}}$ is a nontrivial $\qq$-automorphism of $L$. Let $L_0=L\cap \rr$, so that $L=L_0(\sqrt{-d})$, where $d\in L_0$ and $d>0,$ and let $\lambda\in L_0^\times.$ Then
the symmetric $\qq$-bilinear map $$\nfunc{q_{L,\lambda}}{L\times L}{\qq}{(x,y)}{\Tr_{L/\qq}(\lambda x^*y)}$$ is nondegenerate.

Moreover, $q_{L,\lambda}$ is positive definite if and only if $L$ is a CM field and $\lambda$ is totally positive.
\end{lem}

\begin{proof}
Let us keep the notation of the lemma. We may write $x=x_0+x_1\sqrt{-d}$ and $y=y_0+y_1\sqrt{-d},$ where $x_0,x_1,y_0,y_1\in L_0$.
We have $$\Tr_{L/\qq}(\sqrt{-d})=\Tr_{L_0/\qq}(\Tr_{L/L_0}(\sqrt{-d}))=\Tr_{L_0/\qq}(0)=0.$$ Hence, we get 
$$q_{L,\lambda}(x,y)=\Tr_{L/\qq}(\lambda x_0y_0+d\lambda x_1y_1)=
\Tr_{L_0/\qq}(\Tr_{L/L_0}(\lambda x_0y_0+d\lambda x_1y_1)),$$ that is $$q_L(x,y)=\Tr_{L_0/\qq}(2\lambda x_0y_0)+\Tr_{L_0/\qq}(2d\lambda x_1y_1),$$
from which we have $$q_{L,\lambda}\simeq \mathcal{T}_{L_0,2\lambda}\perp  \mathcal{T}_{L_0,2d\lambda}.$$
Thus, $q_{L,\lambda}$ is  positive definite if and only if its two orthogonal summands are.  Now, apply the previous lemma to conclude that $q_{L,\lambda}$ is positive definite if and only if $L_0$ is totally real, and $d$ and $\lambda$ are totally positive (taking into account that $2$ is totally positive). The fact that $L_0$ is totally real and $d$ is totally positive is equivalent to say that $L$ is a CM field, so we are done.
\end{proof}

Let $L$ be a number field and 
consider the symmetric $\qq$-bilinear form $$\nfuncv{q_{L,\lambda}}{L\times L}{\qq}{(x,y)}{\Tr_{L/\qq}(\lambda x^* y)}$$
where $\lambda\in L\cap\rr$ is totally positive. It is positive definite if $L$ is CM or totally real (in this case, complex conjugation is trivial on $L$) in view of the previous lemmas.

For any fractional ideal $J$ of $\OO_L$, we have $\Tr_{L/\qq}(J)\subset\zz$ if and only if $J\subset \mathcal{D}_L^{-1}$, where $\mathcal{D}_L$ is the different ideal. We then have the following result.

\begin{lem}\label{qL}
Let $L/\qq$ be a number field of degree $n$ and discriminant $d_L.$ Assume that $L$ is CM or totally real, and let $\lambda\in L_0=L\cap \rr$ be totally positive. For any ideal $J$ of $\OO_L$,  its dual $J^\dl$ with respect to $q_{L,\lambda}$ is $$J^\dl=(\lambda J^*\mathcal{D}_L)^{-1},$$
 where $\mathcal{D}_L$ is the different ideal.
 
In particular, $J$ is an integral lattice if and only if $\lambda J^*J\subset \mathcal{D}_L^{-1}$.
In this case, we have $$\det(J)=N_{L/\qq}(\lambda)N_{L/\qq}(J)^2 \vert d_L\vert.$$

\end{lem}


\begin{proof}
By definition of $\mathcal{D}_L$, we have $x\in J^\dl$ if and only if $\lambda x^*J\subset \mathcal{D}_L^{-1}$, that is if and only if $x\in (\lambda J^* \mathcal{D}_L^*)^{-1}$, taking into account that $\lambda^*=\lambda$. Since complex conjugation induces an automorphism of $L$, we have $\mathcal{D}_L^*=\mathcal{D}_L$. Finally, we get $J^\dl=(\lambda J^*\mathcal{D}_L)^{-1},$ hence the first part of the lemma.

Let $\sigma_1,\ldots,\sigma_n$ the embeddings of $L$, and let $(\omega_1,\ldots,\omega_n)$ be a $\zz$-basis of $\OO_L$ adapted to $J$, so that there exist integers $q_1,\ldots,q_n\geq 1$ such that $q_1\mid\cdots\mid q_n$ and $(q_1\omega_1,\ldots,q_n\omega_n)$ is a $\zz$-basis of $J.$ It is known that we have $$ d_L= \det(\Tr_{L/\qq}(\omega_i\omega_j))_{i,j}=\det(W^t W)=\det(W)^2,$$
where $W=(\sigma_i(\omega_j))_{i,j}$. Note that, in the case where $L$ is totally real, $W$ has real entries and $d_L>0$. It follows that $\vert\det(W)\vert^2=\vert d_L\vert$ in both considered cases.

Now, since complex conjugation is a $\qq$-automorphism of $L$ which commutes with any embedding (even in the case where $L$ is totally real, since in this case, complex conjugation is the identity morphism on $L$) and $\sigma_i(\lambda)$ is a positive real number for all $i\in\left\llbracket 1,m\right\rrbracket$, we get easily the equality 
$$\det(J)=\det(\Tr_{L/\qq}(\lambda q_i^*\omega^*_iq_j\omega_j))_{i,j}=(q_1\cdots q_n)^2\det((D^*W^*)^t DW),$$
where $D={\rm diag}(\sqrt{\sigma_1(\lambda)},\ldots,\sqrt{\sigma_n(\lambda)}).$ Notice also that $q_1\cdots q_n=[\OO_L:J]=N_{L/\qq}(J)$.
Therefore $$\det(J)=N_{L/\qq}(J)^2 \det(D)^2\det(W)^*\det(W)=N_{L/\qq}(J)^2 N_{L/\qq}(\lambda)\vert \det(W)\vert^2.$$ Finally, $\det(J)=N_{L/\qq}(\lambda)N_{L/\qq}(J)^2 \vert d_L\vert.$

\end{proof}


We now give a sufficient condition to obtain even integral lattices in the case where $J=\OO_L.$

\begin{lem}\label{evengc}
Assume that $(\OO_L,q_{L,\lambda})$ is an integral lattice, and that the induced form on $\OO_L/I$ is a well-defined nondegenerate $\ff_p$-bilinear form.  Let $C$ be an ideal of $\OO_L/I$  such that $C\subset C^\perp$. Assume that $p$ is odd, and that there exists $u\in\OO_L$ such that  $u^*+u=1$. Then $\Gamma_C$ is even.
\end{lem}

\begin{proof}
Let ${\bf e}=(e_1,\ldots,e_n)$ be a $\zz$-basis of $\pi^{-1}(C)$, so that $\frac{1}{\sqrt{p}}{\bf e}$ is a $\zz$-basis of $\Gamma_C$.
By Theorem \ref{thm-gamma}, $\Gamma_C$ is an integral lattice. If $G$ is  the representative matrix of $(q_{L,\lambda})_\rr$ in the basis $\frac{1}{\sqrt{p}}{\bf e}$, the representative matrix of $q_{L,\lambda}$ in the basis {\bf e} is $pG$. 
For all $x\in \pi^{-1}(C)$, we have $$q_{L,\lambda}(x,x)=\Tr_{L/\qq}(\lambda x^*x)=\Tr_{L/\qq}(\lambda x^*u^* x)+\Tr_{L/\qq}(\lambda x^*u x),$$
that is  $$q_{L,\lambda}(x,x)=q_{L,\lambda}(ux,x)+q_{L,\lambda}(x,ux)=2q_{L,\lambda}(ux,x).$$
Since $(\OO_L,q_{L,\lambda})$ is an integral lattice, we conclude that $q_{L,\lambda}(x,x)\in 2\zz$ for all $x\in \pi^{-1}(C),$ meaning that the diagonal entries of  $pG$ are even.
Since $p$ is odd, it follows easily that the diagonal entries of $G$ are even, that is $\Gamma_C$ is an even integral lattice.
\end{proof}

\begin{lem}\label{lemh1}
Let $K_1, K_2$ be two Galois number fields of degree $n_1$ and $n_2$ respectively which are arithmetically disjoint (that is  linearly disjoint over $\qq$ with coprime discriminants).  Assume that $\mathcal{O}_{K_i}=\zz[\alpha_i]$ for $i=1,2.$
Let $p$ be a prime number which is totally and  tamely ramified in $K_2$, and let $\pp_2$ be the unique prime ideal of $\mathcal{O}_{K_2}$ lying above $p$. Set $L=K_1K_2$. Assume that $L$ is totally real or a CM field, and that complex conjugation induces an automorphism on $K_1$ and on $K_2$.

Let $\lambda_1\in K_1^\times$ be a real totally positive element. Assume that $p$ is coprime to $\lambda_1$,  that $v_\pp(\lambda_1)=0$ for all prime ideals $\pp$ of $\mathcal{O}_{K_1}$ containing $p$, and that $\lambda_1\in\mathcal{D}_L^{-1}$.

Then the bilinear maps 
$$\nfunc{\varphi_1}{\mathcal{O}_{K_1}/p{\mathcal{O}_{K_1}\times \mathcal{O}_{K_1}/p{\mathcal{O}_{K_1}}}}{\ff_p}{([x_1]_p,[y_1]_p)}{[\Tr_{K_1/\qq}(\lambda_1x_1^*y_1)]_p}$$ 
and 
$$\nfunc{\varphi}{\mathcal{O}_L/\pp_2{\mathcal{O}_L\times \mathcal{O}_L/\pp_2\mathcal{O}_L}}{\ff_p}{([x]_{\pp_2},[y]_{\pp_2})}{[\Tr_{L/\qq}(\lambda_1x^*y)]_p}$$
are both  well-defined and nondegenerate.

Moreover, let $$\theta: \mathcal{O}_{K_1}/p\mathcal{O}_{K_1}\overset{\sim}{\to} \mathcal{O}_L/\pp_2\mathcal{O}_L$$ be the isomorphism of Proposition \ref{proph2}, which sends $[x_1]_p$ to $[x_1]_{\pp_2}$. Then for all $x_1,y_1\in\mathcal{O}_{K_1}$, we have  $$\varphi(\theta([x_1]_p),\theta([y_1]_p))=[n_2]_p\varphi_1([x_1]_p,[y_1]_p).$$
\end{lem}

\begin{proof}
By assumption on $K_1$ and $K_2$, we have $\mathcal{D}_L= \mathcal{D}_{K_1}\mathcal{D}_{K_2}$, and then $\mathcal{D}_L^{-1}= \mathcal{D}_{K_1}^{-1}\mathcal{D}_{K_2}^{-1}$.
In particular, $\lambda_1\in \mathcal{D}_L^{-1}\subset\mathcal{D}_{K_1}^{-1}$.
By Lemma \ref{qL}, we have $\Tr_{K_1/\qq}(\lambda_1x_1^*y_1)\in\zz$ for all $x_1,y_1\in\mathcal{O}_{K_1}$, and $\Tr_{L/\qq}(\lambda_1x^*y)\in\zz$ for all $x,y\in\mathcal{O}_L$. Moreover, the determinant of the lattice $\mathcal{O}_{K_1}$ with respect to the first bilinear form is $N_{K_1/\qq}(\lambda_1)\vert d_{K_1}\vert$. Since $p$ is unramified in $K_1$, we have $v_\pp(\vert d_{K_1}\vert)=0$ for all prime ideals $\pp$ of $\mathcal{O}_{K_1}$ containing $p$. The same being true for $\lambda_1$,  it follows that $p$ does not divide $\det(\mathcal{O}_{K_1})$. Lemma \ref{lem-det} (with $M=\OO_{K_1}$ and $N=p\OO_{K_1}$) then shows that $\varphi_1$ is well-defined and nondegenerate.

Now, for all $x_1,y_1\in\mathcal{O}_{K_1}$, we have $$\varphi(\theta([x_1]_p),\theta([y_1]_p))=\varphi([x_1]_{\pp_2},[y_1]_{\pp_2})=[\Tr_{L/\qq}(\lambda_1x_1^*y_1)]_p,$$
and thus $$\varphi(\theta([x_1]_p),\theta([y_1]_p))=[n_2\Tr_{K_1/\qq}(\lambda_1x_1^*y_1)]_p=[n_2]_p\varphi_1(\theta([x_1]_p),\theta([y_1]_p)).$$
Since $p$ is totally and tamely ramified in $K_2$, we have $p\nmid n_2$, so $[n_2]_p$ is not zero. Since $\varphi_1$ is nondegenerate, the same holds for $\varphi$.
\end{proof}

In the situation of Lemma \ref{lemh1}, we then may use Theorem \ref{dualpoly} and Proposition \ref{calculgstar} to provide examples.

\begin{ex}
Consider the cyclotomic fields $K_1=\qq(\zeta_1)$, $K_2=\qq(\zeta_2)$, where $\zeta_1,\zeta_2$ are respectively primitive $p_1^{r_1}$th and $p_2^{r_2}$th roots of unity, where  $p_1$ and $p_2$ are distinct primes. 
Then, by Proposition \ref{proph2}, $\ff_{p_2}[X]/(\mu_{\alpha_1,\qq})\simeq \OO_{K_1}/p_2\OO_{K_1}$ for $\mu_{\alpha_1,\qq}$ the corresponding cyclotomic polynomial. Let $\ov{g}$ be a divisor of $\ov{\mu}_{\alpha_1,\qq}$. 
By Proposition \ref{calculgstar}, since $\alpha_1=\zeta_1$ satisfies $\alpha_1^*\alpha_1 =1$, we have
\[
\ov{g}_* = \ov{g}(0)^{-1}X^{\deg(\ov{g})}\ov{g}(X^{-1}).
\] 
Theorem \ref{dualpoly} tells that the code generated by $g$ is self-orthogonal if and only if $\ov{\mu}_{\alpha_1,\qq}|\ov{g}_*\ov{g}$ and self-dual if and only if $\ov{\mu}_{\alpha_1,\qq}=\ov{g}_*\ov{g}$. Examples of cases where $\ov{\mu}_{\alpha_1,\qq}=\ov{g}_*\ov{g}$ are shown in the table below. In all the following examples, we will take $r_2=1$ and $p=p_2$, so that $p$ will be totally and tamely ramified in $K_2$. 
\end{ex}

\begin{center}
\begin{tabular}{|l|l|}
\hline
$\alpha_1=\zeta_8$, $\ov{\mu}_{\alpha_1,\qq}=X^4+\ov{1}=\ov{g}_*\ov{g}$ & $p=p_2$ \\
\hline
$(X^2 + X + \ov{2})(X^2 + \ov{2}X + \ov{2})$ & $3$\\
$(X^2 + \ov{2})(X^2 + \ov{3})$ & $5$\\
$(X^2 + \ov{3}X + \ov{10})(X^2 + \ov{8}X + \ov{10})$ & $11$\\
$(X^2 + \ov{5})(X^2 + \ov{8})$ & $13$\\
$(X^2 + \ov{6}X + \ov{18})(X^2 + \ov{13}X + \ov{18})$ & $19$\\
\hline
$\alpha_1=\zeta_{16}$, $\ov{\mu}_{\alpha_1,\qq}=X^8+\ov{1}=\ov{g}_*\ov{g}$ & $p=p_2$ \\
\hline
$(X^4 + X^2 + \ov{2})(X^4 + \ov{2}X^2 + \ov{2})$ & $3$\\
$(X^4 + \ov{2})(X^4 + \ov{3})$ & $5$\\
$(X^4 + \ov{3}X^2 + \ov{10})(X^4 + \ov{8}X^2 + \ov{10})$ & $11$\\
\hline
$\alpha_1=\zeta_{9}$, $\ov{\mu}_{\alpha_1,\qq}=X^6+X^3+\ov{1}=\ov{g}_*\ov{g}$ & $p=p_2$ \\
\hline
$(X^3 + \ov{3})(X^3 + \ov{5})$ & $7$\\
$(X^3 + \ov{4})(X^3 + \ov{10})$ & $13$\\
\hline
$\alpha_1=\zeta_{7}$, $\ov{\mu}_{\alpha_1,\qq}=X^6+X^5+\ldots+X+\ov{1}=\ov{g}_*\ov{g}$ & $p=p_2$ \\
\hline
$(X^3 + X + \ov{1})(X^3 + X^2 + \ov{1})$ & $2$\\
$(X^3 + \ov{5}X^2 + \ov{4}X+\ov{10})(X^3 + \ov{7}X^2 + \ov{6}X+ \ov{10})$ & $11$\\
\hline
\end{tabular}
\end{center}

The above examples show codes $C$ built over the alphabet $\ff_{p_1}[X]/(\mu_{\alpha_1,\qq})$, but since 
$\ff_{p_1}[X]/(\mu_{\alpha_1,\qq})\simeq \OO_{K_1}/p\OO_{K_1}$, we are also in the situation of Theorem \ref{thm-gamma}, and we can thus consider the corresponding lattices $\Gamma_C$, using Lemma \ref{lemh1} for conditions on $p$ and $\lambda_1$ to have a suitable nondegenerate symmetric bilinear form.

\begin{ex}
Let $L=\qq(\zeta_{8p}), K_1=\qq(\zeta_8),K_2=\qq(\zeta_{p})$, $p=3$, $5$, $11$, $13$, $19$. We have $\alpha_1=\zeta_8$ and $\alpha_2=\zeta_{p}$. The degree of $L$ is $4(p-1)$. 
Take $\lambda_1=\dfrac{1}{4}$, so that with respect to the chosen symmetric bilinear form 
$$\det(\mathcal{O}_L)=\dfrac{1}{2^{8(p-1)}}(2^{8(p-1)} p^{4(p-2)})=p^{4(p-2)}.$$
We have $\ov{\mu}_{\alpha_1,\qq}=X^4+1$. 
Take $\ov{g}$ to be $X^2+X+\ov{2}\pmod 3$, $X^2 + \ov{2} \pmod 5$, $X^2 + \ov{3}X + \ov{10} \pmod{11}$, 
$X^2 + \ov{5} \pmod{13}$, respectively $X^2 + \ov{6}X + \ov{18} \pmod{19}$.

The corresponding code $C$ of $\mathcal{O}_L/I$ has dimension $2$ over $\ff_p$, hence codimension $2$, since $\mathcal{O}_L/\pp_2\mathcal{O}_L$ has dimension $n_1=4$ . 
As computed in the above example, for these cases, $\ov{g}_*\ov{g}=\ov{\mu}_{\alpha_1,\qq}$ and it follows from Theorem \ref{dualpoly} that $C^\perp=C$.  Moreover, since $n=4(p-1)$ and $C$ has codimension $2$, Theorem \ref{thm-gamma} $(4)$ shows that $\Gamma_C$ is unimodular. 

The corresponding lattice $\Gamma_C$ is $$\Gamma_C=\dfrac{1}{\sqrt{p}}(\mathcal{O}_L \, g(\alpha)+\pp_2).$$
For $p=3$, as $(-\zeta_3)+(-\zeta_3)^*=1$, Corollary \ref{evengc} shows that $\Gamma_C$ is even unimodular, hence isomorphic to $E_8$.

The other constructions are giving unimodular lattices in dimension 16, 40, 48 and 72.

Repeating the same computations with $L=\qq(\zeta_{16p})$ and $\lambda_1=\tfrac{1}{8}$, for $p=3$, $5$, $11$, gives according to the above example unimodular lattices in dimension 16, 32 and 80. 
\end{ex}

We have already seen examples where hypothesis $(H_1)$ is fulfilled. We now investigate the general case.

\begin{lem}\label{modp}
Let $L/\qq$ be a Galois totally real or CM Galois extension, let $\lambda\in L_0=L\cap\rr$ be totally positive and let $p$ be a prime number. 
For $I,J$ ideals of $\OO_L$ such that $pJ\subset I\subset J$ and $\lambda J^*I \subset p\mathcal{D}^{-1}_L$,
the abelian group $J/I$ may be endowed with a canonical structure of an $\ff_p$-vector space, given by $$\funcv{\ff_p\times J/I}{J/I}{([m]_p, [x]_I)}{[mx]_I}$$ and the integral lattice $(J, q_{L,\lambda})$ induces a bilinear form $$\nfuncp{\varphi_{J/I,\lambda}}{J/I\times J/I}{\ff_p}{([x]_I, [y]_I)}{[q_{L,\lambda}(x,y)]_p}$$

The $\ff_p$-bilinear form $\varphi_{J/I,\lambda}$ is nondegenerate if and only if $$(p(\lambda J^*\mathcal{D}_L)^{-1})\cap J=I.$$
\end{lem}
\begin{proof}
The conditions $pJ \subset I$ and $\lambda J^*I \subset p\mathcal{D}^{-1}_L$ imply $p\lambda J^*J \subset \lambda J^* I \subset p\mathcal{D}^{-1}_L$ that is $\lambda J^*J\subset  \mathcal{D}^{-1}_L$ and $J$ is integral by the previous lemma. Furthermore, the condition $\lambda J^*I \subset p\mathcal{D}^{-1}_L$ is equivalent to $q_{L,\lambda}(x,y)\in p\zz$ for all $x\in J$ and all $y\in I$.  Lemma \ref{nondeg} and Lemma \ref{qL} then yield the claim on nondegeneracy.
\end{proof}

\begin{ex}\label{ex:cycl}
Assume that $p$ is totally and tamely ramified in $L$. Set $J=\OO_L$, $I=\pp^j, j\in\left\llbracket 1,m\right\rrbracket$, where  $\pp$ is the unique ideal of $\OO_L$ lying above $p$. We need to check that $p\OO_L \subset \pp^j \subset \OO_L$ and $\lambda\pp^j \subset p\mathcal{D}_L^{-1}$ for a choice of $\lambda \in L\cap\rr$.

First, since $p$ is totally ramified in $L$ of degree $m$, we have $p\OO_L=\pp^m\subset \pp^j \subset \OO_L$ for $j\in\left\llbracket 1,m\right\rrbracket$. We now have that $\lambda p^{-1}\pp^j\subset\mathcal{D}_L^{-1}$ if and only if  
$\lambda p^{-1}\pp^j\mathcal{D}_L\subset\OO_L$.

Since $p$ is tamely ramified, $$v_\pp(\lambda p^{-1}\pp^j\mathcal{D}_L)=v_\mathfrak{p}(\lambda)-e_p+j+(e_p-1)\geq 0 $$ if and only if $v_\mathfrak{p}(\lambda)+j \geq 1$.
If $\mathfrak{q}$ is a prime ideal of $\OO_L$ lying above a prime number $q\neq p$, we have 

$$v_\mathfrak{q}(\lambda p^{-1}\pp^j\mathcal{D}_L)=v_\mathfrak{q}(\lambda)+v_\mathfrak{q}(\mathcal{D}_L)=v_\mathfrak{q}(\lambda)+d_q\geq 0.$$
This proves the desired inclusion for $\lambda\in L\cap\rr$ such that $$v_\mathfrak{p}(\lambda)+j \geq 1,~v_\mathfrak{q}(\lambda)+d_q\geq 0.$$
To check that $\varphi_{J/I,\lambda}$ is nondegenerate, we further need $\lambda$ such that $$p\lambda^{-1} \mathcal{D}_L^{-1} \cap \OO_L = \mathfrak{p}^j,$$ that is we need 
$$\max(v_{\mathfrak{p}}(p\lambda^{-1} \mathcal{D}_L^{-1}),0)=\max(1-v_{\mathfrak{p}}(\lambda),0)=j. $$
Since $v_\mathfrak{p}(\lambda)+j \geq 1$, this is equivalent to ask for $j+v_{\mathfrak{p}}(\lambda)=1$.

If $\mathfrak{q}$ lies above a prime $q\neq p$, we need $$\max((-d_q-v_{\mathfrak{q}}(\lambda),0)=0, $$
which is fulfilled automatically since  $v_\mathfrak{q}(\lambda)+d_q\geq 0.$

To summarize, $\varphi_{J/I,\lambda}$ is well-defined and nondegenerate if and only if $v_\pp(\lambda)=1-j$ and $v_\mathfrak{q}(\lambda)\geq -d_q$ for all prime ideals $\mathfrak{q}\neq\mathfrak{p}$.

We now compute an $\ff_p$-basis of $\OO_L/\pp^j$.

Let $\alpha\in\pp\setminus\pp^2$. Then, $v_\pp(\alpha)=1$. Since $p \in \pp,$ the map $$\nfunc{f}{\ff_p[X]}{\OO_L/\pp^j}{\ov{P}}{[P(\alpha)]_{\pp^j}}$$ is a well-defined morphism of $\ff_p$-algebras.
Let $\ov{P}=\ov{a}_nX^n+\cdots+\ov{a}_0\in\ker(f)$, where $a_0,\ldots,a_n\in\left\llbracket 0, p-1\right\rrbracket$. Since $\alpha^k\in\pp^j$ for all $k\geq j$, we get $a_{j-1}\alpha^j+\cdots+a_0\in\pp^j$.
If $a_k\neq 0$, it is coprime to $p$ , hence coprime to $\pp$, so $v_\pp(a_k\alpha^k)=k$. It easily follows that, if one of the coefficients $a_0,\ldots,a_{j-1}$ is nonzero, then $v_\pp(a_{j-1}\alpha^{j-1}+\cdots+a_0)=k$, where $k$ is the smallest integer such that $a_k\neq 0$. In particular, this valuation is $\leq j-1$, which contradicts the fact that $a_{j-1}\alpha^j+\cdots+a_0\in\pp^j$. Hence, $a_0,\ldots,a_{j-1}$ are all zero, and $\ov{P}\in(X^j)$. The converse being clear, we get $\ker(f)=(X^j)$. Consequently, $f$ induces an injective morphism $\ff_p[X]/(X^j)\to \OO_L/\pp^j$ of $\ff_p$-algebras, which is an isomorphism since $$\vert\OO_L/\pp^j\vert =N_{L/\qq}(\pp)^j=p^j= \vert \ff_p[X]/(X^j)\vert.$$
It follows that 
an $\ff_p$-basis of $\OO_L/\pp^j$ is $[1]_{\pp^j},[\alpha]_{\pp^j},\ldots,[\alpha^{j-1}]_{\pp^j}$. 
Hence $\varphi_{J/I,\lambda}$ is nondegenerate of rank $j$.


A concrete example of this situation is $L=\qq(\zeta_p)$ for $\zeta_p$ a primitive $p$-th root of unity (where $p$ is an odd prime). In this case,  $\pp=(1-\zeta_p)$, and one may take $\alpha=1-\zeta_p$.

Since $p$ is the unique prime number which ramifies, the required conditions are $v_\pp(\lambda)=1-j$ and $v_\mathfrak{q}(\lambda)\geq 0$ for all prime ideals $\mathfrak{q}\neq\pp$.

If $j=1$, one may take $\lambda=1$, which is obviously a totally positive real number satisfying the conditions above. 
\end{ex}

%
%
%
\section{Skew-Polynomial Codes and Lattices over Cyclic Algebras}

\subsection{Generalities about Skew-polynomial Rings.} 

In this subsection, we define skew-polynomial rings and study briefly their properties. The following considerations are certainly well-known, but since we have a  definition of skew-polynomial rings which slightly differ from the usual one (in order to be more suitable for our purpose), we prefer to give proofs of the results below for sake of completeness.

Let $K$ be a field, and let $\sigma$ be a ring automorphism of $K$. For all $a\in K$, we will set $a^\sigma=\sigma^{-1}(a).$

 The skew-polynomial ring $K[X;\sigma]$ is the set of formal sums $$\sum_{n\geq 0} X^n a_n, \ a_n\in K,$$ where only finitely many $a_n's$ are nonzero, endowed with the following operations :
\begin{eqnarray*}
\sum_{n\geq 0} X^n a_n+\sum_{n\geq 0} X^n b_n & = & \sum_{n\geq 0} X^n (a_n+b_n), \\
(\sum_{n\geq 0} X^n a_n)(\sum_{n\geq 0} X^n b_n) &= &\sum_{n,m\geq 0} \ X^{n+m} a_n^{\sigma^m}b_m \\
 &=&\sum_{n\geq 0} \ X^n (\sum_{k=0}^n(a_k^{\sigma^{n-k}}b_{n-k}).
\end{eqnarray*}
If $f\in K[X;\sigma]\setminus\{0\}$, the degree of $f=\ds\sum_{n\geq 0} X^n a_n$, denoted by $\deg(f)$, is the greatest integer $n\geq 0$ such that $a_n\neq 0$. We also set $\deg(0)=-\infty$. 
Since $\sigma$ is an automorphism, we have as usual $$\deg(fg)=\deg(f)+\deg(g) \ \mbox{ for all }f,g\in K[X;\sigma].$$

In particular, $K[X;\sigma]$ has no zero divisors.

The ring $K[X;\sigma]$ is left Euclidean and right Euclidean for the degree function. More precisely, for all $f,g\in K[X;\sigma]$, with $g\neq 0$, there exist unique polynomials $Q_l,R_l, Q_r,R_r\in K[X;\sigma]$ satisfying $$f=g Q_l+R_l, \deg(R_l)<\deg(g) \ \mbox{ and } \ f=Q_r g+R_r, \deg(R_r)<\deg(g).$$

Indeed, write $f=\ds\sum_{k=0}^n X^k a_k$ and $g=\ds\sum_{k= 0}^m X^k b_k$, with $b_m\in K^\times$. If $n<m$, take $Q_l=Q_r=0$ and $R_l=R_r=g$.

Assume now that $n\geq m$. Then one may check that the polynomials $f-g\cdot (X^{n-m}(b_m^{-1})^{\sigma^{n-m}}a_n)$ and $f-(X^{n-m} (a_nb_m^{-1})^{\sigma^{-m}})\cdot g$ both have degree $<n$.

By induction, we get the existence of $Q_l,Q_r,R_l,R_r$, as in the usual case (corresponding to $\sigma=\Id_K$). The uniqueness part follows from degree considerations.

It follows that any  left ideal or right ideal of $K[X;\sigma]$ is principal. More precisely, a nonzero left or right ideal is generated by its unique monic element of smallest degree.

\begin{rem}\label{rem-divisor}
A left divisor of a given polynomial is not necessarily a right divisor. However, everything becomes natural for central polynomials, that is polynomials lying in the center of $K[X;\sigma]$.

Indeed, let $f\in K[X;\sigma]$ be a nonzero central polynomial, and let $g,h\in K[X;\sigma]$ be polynomials satisfying $f=gh$. Then $$hf=(hg)h=fh=gh^2,$$ so $(hg-gh)h=0$.
Since $f$ is nonzero, $h$ is nonzero, and since $K[X;\sigma]$ has no zero divisors, we get $hg=gh$. Hence $g$ and $h$ are both left and right divisors of 
$f$, and we may speak of divisors in this case.
\end{rem}

We end this section with considerations on morphisms. Let $R$ be a ring (with $1$). Let $u:K\to R$ be a ring morphism and $r\in R$ satisfying the relations $$u(a)r=r(u(a))^\sigma \ \mbox{ for all } a\in K.$$ Then there exists a unique ring morphism $\psi_{u,r}:K[X;\sigma]\to R$ sending $X$ to $r$ and whose restriction to $K$ is $u$. It is defined by $$\psi_{u,r}(\sum_{n\geq 0}X^n a_n)=\sum_{n\geq 0}r^n u(a_n)$$ for all $a_n\in K$ almost all zero.


\subsection{The context.}

Let $L/k$ be a cyclic Galois number field extension of group $G=\langle \sigma \rangle$ and degree $n$. 
Assume that  $L/\qq$ is totally real or CM, and that complex conjugation induces a ring automorphism on $k$ (possibly trivial).

Let $\gamma\in k^\times$ such that $\gamma\gamma^*=1$, and consider the cyclic algebra $B=(\gamma,L/k,\sigma)$, which can be written in a canonical basis as
\[
B = L \oplus eL \oplus \cdots \oplus e^{n-1}L = \bigoplus_{j=0}^{n-1} e^j L
\]

where $e^n=\gamma$ and $a e=e a^\sigma $ for all $a\in L$. Note that $e\in B^\times$, and that $e^{-1}=e^{n-1}\gamma^{-1}$.

We define a map $\tau:B\to B$ by setting $$\tau(\sum_{j=0}^{n-1}e^j x_j)=\sum_{j=0}^{n-1}x_j^*e^{-j}$$ for all $x_j\in L$.
This map is an involution on $B$, that is, an anti-automorphism of order two (see \cite[Lemma IX.4.3 and Example IX.4.5]{BOCSA}, where the involution is supposed to be nontrivial on $k$. 
However, the computations are the same in our context).

Recall now that we have a nonzero $k$-linear map $\Trd_B:B\to k$, called the reduced trace of $B$, that can be computed in our case as follows  (this follows from the definition of the reduced trace, and from \cite[Lemma VI.3.1]{BOCSA}):

\[
\Trd_B(\sum_{j=0}^{n-1}e^jx_j)=\Tr_{L/k}(x_0).
\]

\begin{lem}\label{comptrd}
Let $x,y\in B$, and let $\lambda\in L$. Then, we have $$\Trd_B(\lambda \tau(x)y)=\sum_{j=0}^{n-1}\Tr_{L/k}(\lambda x^*_j y_j),$$
where $x=\ds\sum_{j=0}^{n-1}e^j x_j$ and  $y=\ds\sum_{j=0}^{n-1}e^j y_j.$
\end{lem}

\begin{proof}
We have 
$\tau(x)=\ds\sum_{j=0}^{n-1}x^*_j e^{-j},$
and therefore $$\lambda\tau(x)y=\sum_{i,j=0}^{n-1}\lambda x^*_ie^{-i}e^j y_j.$$

In view of the definition of the product on $B$,  the only terms which contribute to the constant term are those corresponding to $i=j$.

Thus, $$\Trd_B(\lambda\tau(x)y)=\sum_{j=0}^{n-1}\Tr_{L/k}(\lambda x^*_j y_j),$$
hence the result.
\end{proof}

We then have the following result.

\begin{lem}\label{lemqb}
Let $\lambda\in L\cap \rr$, $\lambda \neq 0$.
Then the  bilinear map $$\nfunc{q_{B,\lambda}}{B\times B}{\qq}{(x,y)}{\Tr_{k/\qq}(\Trd_B(\lambda\tau(x)y))}$$ is symmetric and nondegenerate.

If $x=\ds\sum_{j=0}^{n-1}e^j x_j$ and  $y=\ds\sum_{j=0}^{n-1}e^j y_j,$ then 

$$q_{B,\lambda}(x,y)=\sum_{j=0}^{n-1}\Tr_{L/\qq}(\lambda x^*_j y_j).$$

Moreover, $q_{B,\lambda}$ is positive definite if and only if 
$L/\qq$ is totally real or CM, and $\lambda$ is totally positive.
\end{lem}

\begin{proof}
Lemma \ref{comptrd} shows the desired equality, and the symmetry comes from the fact that $\lambda^*=\lambda$ and the properties of the trace. It also follows from the equality above that $$q_{B,\lambda}\simeq n\times q_{L,\lambda},$$ where $n=\deg(B).$ 
We now may apply Lemmas \ref{traceprops} and \ref{trstarprops} to conclude.
\end{proof}

From now on, we assume that $L/\qq$ is a Galois totally real or CM extension, and that $\lambda\in L_0^\times$ is totally positive.

Now that we have a positive definite symmetric bilinear form on $B$, we would like to build an integral lattice from it.

For, we impose a stronger condition on $\gamma$, namely $\gamma\in\OO_k$. Since $\gamma \gamma^*=1$, this implies that $\gamma\in\OO_k^\times$.

We set $$\Lambda=\bigoplus_{j=0}^{n-1}e^j\OO_L.$$

Consider the subring $\Lambda =\ds \bigoplus_{j=0}^{n-1} e^j \OO_L$. Notice that $e^{-1}=e^{n-1}\gamma^{-1}\in\Lambda$. This equality easily implies that $\Lambda$ is a subring of $B$ satisfying $\tau(\Lambda)=\Lambda$. Moreover, if we suppose that $\lambda\in L_0$ is a totally positive element such that $\lambda\in\mathcal{D}_L^{-1}$, the previous results ensure that $q_{B,\lambda}(x,y)\in \zz$ for all $x,y\in\Lambda$, so $\Lambda$ is a full integral lattice of $(B,q_{B,\lambda})$.   Moreover, a reasoning similar to the proof of Lemma \ref{lemqb} shows that this lattice is isomorphic to $n$ copies of the lattice $q_{L,\lambda}:\OO_L\times\OO_L\to \zz$. In particular, we $$\det(\Lambda)=(N_{L/\qq}(\lambda)\vert d_L\vert)^n.$$

\subsection{Quotients of Orders in Cyclic Algebras.}

Our next goal is to get  results similar  to the results of Subsection \ref{ssec-codes} in our noncommutative context.

Let us repeat one more time our assumptions: $L/k$ is a cyclic extension of number fields, whose Galois group is generated by $\sigma$, complex conjugation induces a ring automorphism on $k$, $\gamma\in\OO_k$ satisfies $\gamma\gamma^*=1$, and $\lambda\in L_0=L\cap\rr$ is a totally positive element such that $\lambda\in\mathcal{D}_L^{-1}$.

We then set $B=(\gamma,L/k,\sigma),$ and $\Lambda=\ds\bigoplus_{j=0}^{n-1}e^j\OO_L.$
Finally, $\tau$ is the involution on $B$ defined by $$\tau(\sum_{j=0}^{n-1}e^j x_j)=\sum_{j=0}x_j^*e^{-j}$$ for all $x_j\in L$.

Let $\pp$ be a prime ideal of $\OO_k$ lying above a prime $p$ satisfying $\pp=\pp^*$, and which is inert or totally ramified in $L$. In this case, there is a unique prime ideal $\PP$ of $\OO_L$ lying above $\pp$. Hence, $\PP^\sigma=\PP$ and $\PP^*=\PP$.

The left ideal $\mathcal{P}=\Lambda \PP$ is then a two-sided ideal of $\Lambda$ satisfying $\tau(\mathcal{P})=\mathcal{P}$.

Indeed, for all $a\in \OO_L$, all $x\in \PP$ and all $k\in\left\llbracket 0,n-1\right\rrbracket$, we have $$x(e^ka)=e^kx^{\sigma^k}a\in e^k \PP, $$so 
$\PP\Lambda\subset \Lambda \PP,$ which is enough to prove that $\mP$ is a two-sided ideal of $\Lambda$.
Keeping the notation above, we also have $$\tau(e^k x)=x^*e^{-k}\in\PP^*\Lambda=\PP\Lambda\subset\Lambda \PP,$$
which is again enough to prove that $\tau(\mathcal{P})\subset\mathcal{P}$; the equality then follows by applying $\tau$  to this inclusion.

 The equalities $\PP^\sigma=\PP$ and $\PP^*=\PP$ imply that $\sigma$ and complex conjugation both induce automorphisms on $\OO_L/\PP$, that we will denote respectively by $\ov{\sigma}$ and $*$. The equality $\tau(\mP)=\mP$ implies that $\tau$ induces an anti-automorphism $\ov{\tau}$ of order $1$ or $2$ on the quotient ring $\Lambda/\mP$.

Assuming further that $\lambda \PP\subset p\mathcal{D}_L^{-1}$, and that $p(\lambda\mathcal{D}_L)^{-1}\cap \OO_L=\PP,$ Lemma \ref{modp} shows that the $\ff_p$-bilinear map $$\nfunc{\varphi_{L,\lambda}}{\OO_L/\PP\times\OO_L/\PP}{\ff_p}{([x]_\PPP,[y]_\PPP)}{[\Tr_{L/\qq}(\lambda x^*y)]}$$ is well-defined and nondegenerate.

Lemma \ref{comptrd} and the previous considerations show that the map

$$\nfunc{\varphi_{\Lambda/\mathcal{P},\lambda}}{\Lambda/\mathcal{P}\times \Lambda/\mathcal{P}}{\ff_p}{
          ([x]_\mPP,[y]_\mPP)}{[\Tr_{k/\qq}(\Trd_B(\lambda\tau(x)y))]_p}$$
is well-defined and nondegenerate, which is a condition similar to hypothesis $(H_1)$.

We then have results on parity of integral lattices, as in the commutative case.

\begin{lem}\label{evenlj}
Assume that the conditions above are satisfied. Let $\mJ$ be a left ideal of $\Lambda/\mP$ such that $\mJ\subset \mJ^\perp$. Assume that $p$ is odd, and that there exists $u\in\OO_L$ such that $u^*+u=1$. Then $\Gamma_\mJ$ is an even integral lattice.
\end{lem}

\begin{proof}
Note that for any $x\in\mJ$, we have 
$$q_{B,\lambda}(x,x)=\Tr_{k/\qq}(\Trd_B(\lambda\tau(x)u^*x))+   \Tr_{k/\qq}(\Trd_B(\lambda\tau(x)ux)).$$
Since $u^*=\tau(u)$, we get $$q_{B,\lambda}(x,x)=q_{B,\lambda}(ux,x)+q_{B,\lambda}(x,ux)=2q_{B,\lambda}(ux,x).$$
Now, we may finish the proof as in the proof of Lemma \ref{evengc}.
\end{proof}

We now check that taking the orthogonal with respect to $\varphi$ preserves left ideals.






\begin{lem}\label{lem:Ilrd}
Let $\mJ$ be a left ideal of $\Lambda/\mP$. 
Then
\[
\mJ^\perp = \{ [y]_\mPP \in \Lambda/\mP\mid \varphi_{\Lambda/\mP,\lambda}([x]_\mPP,[y]_\mPP) = 0 \mbox{ for all } [x]_\mPP \in \mJ \}
\]
is a left ideal of $\Lambda/\mP$. 
\end{lem}

\begin{proof}
Write $\mJ=J/\mP,$ where $J$ is a left ideal of $\Lambda$ containing $\mP$.

Given $[y_1]_\mPP,[y_2]_\mPP\in \mJ^\perp$, for all $x\in J$, we have
$$\begin{array}{lll}
\varphi_{\Lambda/\mP,\lambda}([x]_\mPP,[y_1]_\mPP-[y_2]_\mPP) &= & [\Tr_{k/\qq}(\Trd_B(\lambda\tau(x)(y_1-y_2)))]_p\cr 
& = &
\varphi_{\Lambda/\mP,\lambda}([x]_\mPP,[y_1]_\mPP)-\varphi_{\Lambda/\mP,\lambda}([x]_\mPP,[y_2]_\mPP)\cr 
& = &  [0]_p. \end{array}$$

Moreover, for all $[a]_\mPP\in \Lambda/\mP,$ and all $[y]_\mPP\in\mJ^\perp$, we have 

$$\begin{array}{lll}
\varphi_{\Lambda/\mP,\lambda}([x]_\mPP,[a]_\mPP[y]_\mPP) & =& [\Tr_{k/\qq}(\Trd_B(\lambda\tau(x)ay))]_p
\cr &=& [\Tr_{k/\qq}(\Trd_B(\lambda\tau(\tau(a)x)y))]_p \cr
 & = & \varphi_{\Lambda/\mP,\lambda}([\tau(a)]_\mPP [x]_\mPP,[y]_\mPP)\cr 
 & =&  [ 0 ]_p
\end{array}$$
for all $x\in J$, since $\mP$ is a left ideal. Thus, $[a]_\mPP[y]_\mPP\in \mJ^\perp$, and $\mJ^\perp$ is a left ideal of $\Lambda/\mP$.
\end{proof}

We now take care of hypothesis $(H_2)$. Set $\ff_q=\OO_L/\PP$. Let us note for later use that $X^n-[\gamma]_\PPP \in\ff_q[X;\ov{\sigma}]$ is a central polynomial. Indeed, for all $[a]_\PPP\in\ff_q$ and all $k\geq 0$, we have $$(X^k [a]_\PPP)(X^n-[\gamma]_\PPP)=X^{k+n}[a]_\PPP^{\ov{\sigma}^n}-X^k
[a]_\PPP[\gamma]_\PPP.$$
Now, since $\sigma^n=\Id_L$, we also have $\ov{\sigma}^n=\Id_{\ff_q}$. Moreover, since $\gamma\in k$, $[\gamma]_\PPP$ commutes with all the elements of $\ff_q$. Hence we get  $$(X^k [a]_\PPP)(X^n-[\gamma]_\PPP)=X^{k+n}[a]_\PPP-X^k[a]_\PPP[\gamma]_\PPP=X^{k+n}[a]_\PPP-X^k[\gamma]_\PPP[a]_\PPP.$$
But now, we have $$(X^n-[\gamma]_\PPP)(X^k[a]_\PPP)=X^{n+k}[a]_\PPP-[\gamma]_\PPP( X^k[a]_\PPP)=X^{k+n}[a]_\PPP-X^k[\gamma]_\PPP^{\ov{\sigma}^k}[a]_\PPP.$$
Now, since $\gamma\in k$ and $\sigma$ is $k$-linear, we have $[\gamma]_\PPP^{\ov{\sigma}^k}=[\gamma^{\sigma^k}]_\PPP=[\gamma]_\PPP$, and we finally get the equality 
$$(X^k [a]_\PPP)(X^n-[\gamma]_\PPP)= (X^n-[\gamma]_\PPP)(X^k [a]_\PPP)$$ for all $\ov{a}\in\ff_q$ and all $k\geq 0$, which is enough to prove that $X^n-[\gamma]_\PPP$ commutes with any element of $\ff_q[X;\ov{\sigma}]$.

In particular, by Remark \ref{rem-divisor}, left and right divisors coincide. It also implies that the left and right quotients and remainders coincide as well, and that the set of multiples of $X^n-[\gamma]_\PPP$ is a two-sided ideal of $\ff_q[X;\ov{\sigma}]$.

We then have the following lemma, which is similar to hypothesis $(H_2)$.

\begin{lem}\label{lem:iso}
Set $\ff_q=\OO_L/\PP$. There is a unique ring morphism $\rho:\ff_q[X;\ov{\sigma}]\to\Lambda/\mP$ sending $X$ to $[e]_\mPP$ and an element $[a]_\PPP\in\ff_q$ to $[a]_\mPP$, and it induces an isomorphism of $\ff_p$-algebras
\[
\ff_q[X;\ov{\sigma}]/(X^n-[\gamma]_\PPP) \simeq \Lambda/\mP.
\]
\end{lem}

\begin{proof}
Let us prove the existence and uniqueness of $\rho$. The composition of the inclusion $\OO_L\subset \Lambda$ with the canonical projection $\Lambda\to\Lambda/\mP$ yields a ring morphism $\OO_L\to\Lambda/\mP$ whose kernel contains $\PP$ since $\PP\subset \mP=\Lambda\PP$. Hence, we get a ring morphism
 $u:\ff_q\to \Lambda/\mP$ sending $[a]_\PPP\in\ff_q$ to $[a]_\mPP$. For all $[a]_\PPP\in\ff_q$, we have $$u([a]_\PPP)[e]_\mPP=[ae]_\mPP=[ea^\sigma]_\mPP=[e]_\mPP u([a^\sigma]_\PPP)=[e]_\mPP u([a]_\PPP)^{\ov{\sigma}}).$$

By the considerations on skew polynomial rings recalled in a previous section, we get the desired result.
Now assume that $\ov{f}\in\ff_q[X;\sigma]$ satisfies $\rho(\ov{f})=[0]_\mPP$. We may write $\ov{f}=(X^n-[\gamma]_\PPP)\ov{g}+\ov{r}$, where $\deg(r)<n$. Since $e^n=\gamma$, applying $\rho$ yields $\rho(\ov{r})=[0]_\mPP$. 

Write $\ov{r}=\ds\sum_{j=0}^{n-1}X^j[a_j]_\PPP$. Then $[\ds\sum_{j=0}^{n-1}e^j a_j]_\mPP=[0]_\mPP$, that is $\ds\sum_{j=0}^{n-1}e^j a_j\in\mP=\Lambda\PP$.
It is easy to conclude that each $a_j$ lies in $\PP$, which implies in turn that $\ov{r}=0$. Hence, $\ov{f}$ is a  multiple of $X^n-[\gamma]_\PPP$, and we are done using the first isomorphism theorem, since all the maps are clearly $\ff_p$-linear.
\end{proof}

The isomorphism of Lemma \ref{lem:iso} gives us a correspondence between the monic divisors $\ov{g}$ of $X^n-[\gamma]_\PPP$ and the left/right ideals of $\Lambda/\mP$, that is the left/right ideals of $\Lambda$ containing $\mP$.

If $\ov{g}$ is such a divisor, the corresponding left ideal of $\Lambda/\mP$ is the left ideal generated by $[g(e)]_{\mP},$ where $g\in\OO_L[X]$ is any polynomial whose reduction modulo $\PP$ is $\ov{g}$, and the corresponding left ideal of $\Lambda$ is $\Lambda g(e)+\mP$.

To end this paragraph, we would like to understand how behaves the antiautomorphism $\ov{\tau}$ of $\Lambda/\mP$ via this isomorphism with respect to ideals. More precisely, we have the following lemma.

\begin{lem}\label{taug}
Let $\ov{g}=\ds\sum_{k=0}^dX^k[a_k]_\PPP$ be a monic divisor of $X^n-[\gamma]_\PPP$ of degree $d$, and let $\mJ=\Lambda/\mP \cdot [g(e)]_\mPP$ be the corresponding left ideal of $\Lambda/\mP$. Then $\ov{\tau}(\mJ)$ is the right ideal of $\Lambda/\mP$ corresponding to the monic divisor of degree $d$ $$\ov{g}_{\ov{\tau}}=\sum_{k=0}^dX^k[a^*_{d-k}]_\PPP^{\ov{\sigma}^k} [a_0^{*-1}]_\PPP^{\ov{\sigma}^d}.$$
\end{lem}

\begin{proof}
Notice first that the result trivially holds for $d=0$ or $d=n$. Indeed, if $d=0$, we have $\ov{g}=\ov{g}_{\ov{\tau}}=[1]_\PPP$ and $\mJ=\Lambda/\mP=\ov{\tau}(\mJ)$. If $d=n$, we have $\ov{g}=X^n-[\gamma]_\PPP$ and $\mJ=(0)=\ov{\tau}(\mJ)$. Notice now that we have $$\ov{g}_{\ov{\tau}}=X^n-[\gamma^{*-1}]_\PPP=X^n-[\gamma]_\PPP=\ov{g},$$ since $\gamma^*\gamma=1$, so the result is indeed valid in these two cases.

We now assume that $d\in\left\llbracket 1,n-1\right\rrbracket$.
Notice that $\ov{\tau}(\mJ)$ is the right ideal generated by $[\tau(g(e))]_\mPP$. Since $\gamma\in\OO_k^\times$, $[\gamma]_\PPP$ is nonzero in $\OO_L/\PP$, hence invertible. Thus $[e]_\mPP$ is invertible in $\Lambda/\mP$, and $
\ov{\tau}(\mJ)$ is also generated by $[\tau(g(e))e^d]_\mPP$.

Now write $\ov{g}\ov{h}=X^n-[\gamma]_\PPP$, where $\ov{h}=\ds\sum_{k=0}^{n-d}X^k[b_k]_\PPP$ is monic of degree $n-d$. Then $[g(e)]_\mPP[h(e)]_\mPP=[0]_\mPP$. Applying $\ov{\tau}$ yields $[\tau(h(e))]_\mPP[\tau(g(e))]_\mPP=[0]_\mPP$, and thus 

$$[e^{n-d}\tau(h(e))]_\mPP[\tau(g(e))e^d]_\mPP=[0]_\mPP.$$
Now we have $$\tau(g(e))e^d=\ds\sum_{k=0}^da_k^* e^{d-k}=\sum_{k=0}^d a_{d-k}^* e^k=\sum_{k=0}^d e^k (a_{d-k}^*)^{\sigma^k},$$
while $$e^{n-d}h(e)=\ds\sum_{k=0}^{n-d}e^{n-d}b_k^* e^{-k}=\sum_{k=0}^{n-d}e^{n-d-k}(b_k^*)^{\sigma^{-k}}= \sum_{k=0}^{n-d}e^{k}(b_{n-d-k}^*)^{\sigma^{-(n-d-k)}},$$
that is$e^{n-d}h(e)=\ds\sum_{k=0}^{n-d}e^{k}(b_{n-d-k}^*)^{\sigma^{d-k}}.$

The equality $[e^{n-d}\tau(h(e))]_\mPP[\tau(g(e))e^d]_\mPP=[0]_\mPP$ then yields that $X^n-[\gamma]_\mPP$ divides $\tilde{h}\tilde{g}$, where 
$$\tilde{h}=\sum_{k=0}^{n-d}X^{k}[b_{n-d-k}^*]_\PPP^{\ov{\sigma}^{d-k}} \ \mbox{ and } \ \tilde{g}=\sum_{k=0}^d X^k [a_{d-k}^*]_\PPP^{\ov{\sigma}^k}.$$

Notice that the equality $\ov{g}\ov{h}=X^n-[\gamma]_\PPP$ implies that $[a_0]_\PPP[b_0]_\mPP=-[\gamma]_\mPP$. In particular, both $[a_0]_\PPP$ and $[b_0]_\mPP$  are invertible in $\OO_L/\PP$, and  $[b_0]_\mPP=-[\gamma]_\PPP [a_0]_\mPP^{-1}$. It also implies that $\tilde{h}$ has degree $n-d$, with leading coefficient $\ov{b}_0^*$, while $\ov{g}$ has degree $d$, with leading coefficient $[a_0^*]_\PPP^{\ov{\sigma}^d}$.

Hence $\ov{h}\ov{g}$ has degree $n$. Comparing leading terms then yields $$\ov{h}\ov{g}=(X^n-[\gamma]_\PPP)[b_0^*]_\PPP [a_0^*]_\PPP^{\ov{\sigma}^d}.$$
It follows easily that $\ov{g}_{\ov{\tau}}=\tilde{g}[a_0^{*-1}]_\PPP^{\ov{\sigma}^d}$ is a monic right divisor of $X^n-[\gamma]_\PPP$. 
Since they differ by a unit, $\ov{g}_{\ov{\tau}}$ and $\tilde{g}$ correspond to the same right ideal.

By definition of $\tilde{g}$, its generates the same right ideal as $[\tau(g(e))e^d]_\mPP,$ which is $\ov{\tau}(\mJ)$, as we have seen before.
The result follows.
\end{proof}

\subsection{Skew-polynomial Codes.}

As we did in the commutative case, given a  monic divisor of $X^n-\ov{\gamma}$ corresponding to an ideal $\mJ$ of $\Lambda/\mP$, we compute the divisor corresponding to $\mJ^\perp$.

Of course, we assume that the induced bilinear map $\varphi_{\Lambda/\mP,\lambda}$ on $\Lambda/\mP$ is well-defined and nondegenerate. We have seen that this is equivalent to ask that
$\lambda \PP\subset p\mathcal{D}_L^{-1}$, and that $p(\lambda\mathcal{D}_L)^{-1}\cap \OO_L=\PP.$

Contrary to the commutative case, the result may depend on $\lambda$, depending on the ramification of $L/k$, and we will prove different results according to the situation. We start with two lemmas.

\begin{lem}\label{lemdeg}
Write $q=p^d$. Let $\ov{g}$ be  a monic divisor of $X^n-[\gamma]_\PPP$, and let $\mJ$ the corresponding left ideal of $\Lambda/\mP$. Then we have $$\codim_{\ff_p}(\mJ)=d\deg(\ov{g}).$$ 
In particular, if $\ov{h}$ is the monic divisor corresponding to $\mJ^\perp$, then $$\deg(\ov{g})+\deg(\ov{h})=n.$$
\end{lem}

\begin{proof}
Using division of polynomials, it is clear that for any monic polynomial $\ov{f}\in \ff_q[X;\ov{\sigma}]$, $(1,X,\ldots,X^{\deg(f)-1})$ is a basis of $\ff_q[X;\ov{\sigma}]/(\ov{f})$ as a right $\ff_q$-vector space. Hence, we get $$\dim_{\ff_p}(\ff_q[X;\ov{\sigma}]/(\ov{f}))=d\deg(\ov{f}).$$

In particular, $\dim_{\ff_p}(\Lambda/\mP)=dn$. Moreover, we have the following isomorphisms of $\ff_p$-algebras: 
$$(\Lambda/\mP)/\mJ\simeq (\ff_q[X;\ov{\sigma}]/(X^n-\ov{\gamma}))/((\ov{g})/(X^n-\ov{\gamma}))\simeq \ff_q[X;\ov{\sigma}]/(\ov{g}),$$
and thus, it yields $$\dim_{\ff_p}((\Lambda/\mP)/\mJ))=\codim_{\ff_p}(\mJ)=d\deg(\ov{g}).$$
Consequently, we get the first part of the lemma.

For the second, the non-degeneracy of $\varphi_{\Lambda/\mP,\lambda}$ implies that $$\codim_{\ff_p}(\mJ)+\codim_{\ff_p}(\mJ^\perp)=\dim_{\ff_p}(\Lambda/\mP)=dn.$$
Applying the first part and dividing by $d$ then yield the desired equality.
\end{proof}


\begin{lem}\label{lem-trmp}
We have the inclusion $$\Tr_{L/\qq}(\Trd_B(\lambda \mP))\subset p\zz.$$
\end{lem}

\begin{proof}
Let $x\in\mP$, that is $x=\ds\sum_{i=0}^{n-1} e^i x_i, x_i\in\PP$. Then $\Tr_{k/\qq}(\Trd_B(\lambda x))=\Tr_{L/\qq}(\lambda x_0)\in\Tr_{L/\qq}(\lambda \PP)$.
Since $\lambda\PP\subset p\mathcal{D}_L^{-1}$,  we get $$\Tr_{L/\qq}(\lambda \PP)\subset p\Tr_{L/\qq}(\mathcal{D}_L^{-1})\subset p\zz.$$
\end{proof}

We now state and prove the first main result of this section.

\begin{prop}\label{proplram}
Let $\pp$ be a prime ideal of $\OO_K$ lying above $p$ such that $\pp^*=\pp$, and which totally ramifies in $L$, and let $\PP$ be the unique prime ideal of $\OO_L$ lying above $\pp$. Let $\bar{g}=\ds \sum_{i=0}^d X^i [a_i]_\PPP$ be a monic divisor of $X^n-[\gamma]_\PPP,$ and let $\mJ$ be the corresponding left ideal of $\Lambda/\mP$.

 Then $\mJ^\perp$ corresponds to the monic divisor 
$$\bar{h}=\frac{X^n-[\gamma]_\PPP}{\ov{g}_{\ov{\tau}}},$$
where $\ov{g}_{\ov{\tau}}=\ds\sum_{i=0}^dX^k[a^*_{d-i}]_\PPP [a_0^*]_\PPP^{-1}.$
In particular, $\mJ$ is self-orthogonal if and only if   $\ov{g} \, \ov{g}_{\ov{\tau}}$ is a multiple of $X^n-[\gamma]_\PPP$,
and self-dual if and only if $\ov{g} \ \ov{g}_{\ov{\tau}}=X^n-[\gamma]_\PPP.$
\end{prop}

\begin{proof}
We first prove the following preliminary claim.

{\bf Claim. }The $\ff_p$-algebra $\Lambda/\mP$ is commutative. In particular, for all $x_1,x_2,x_3\in \Lambda$, we have $$[\Tr_{k/\qq}(\Trd_B(\lambda x_1x_2x_3))]_p=[\Tr_{k/\qq}(\Trd_B( \lambda x_2x_1x_3))]_p.$$

Since $\mP^\sigma=\mP$, it follows that $\sigma$ lies in the decomposition group ${\rm Dec}(\mP)$ of $\mP$. By assumption on $\mP$, the ramification index is $1$, meaning that the canonical surjective map ${\rm Dec}(\mP)\to \Gal(\OO_L/\PP /\OO_k/\pp)$ is trivial. In other words, $\sigma$ induces the identity map on $\ff_q$, that is $\ov{\sigma}=\Id_{\ff_q}$.

In particular, $\ff_q[X;\ov{\sigma}]=\ff_q[X]$ is commutative, and so is $\Lambda/\mP$ in view of the isomorphism of Lemma \ref{lem:iso}. 
This means that for all $x_1,x_2\in\Lambda$, we have $x_1x_2=x_2x_1+z,$ for some $z\in\mP$.
Since $\mP$ is a two-sided ideal, $zx_3\in\mP$. Thus, the rest of the claim comes from Lemma \ref{lem-trmp}.

Keeping notation of the proposition, we have $\mJ=\Lambda/\mP \cdot [g(e)]_\mPP$, and $\mJ^\perp=\Lambda/\mP \cdot [h(e)]_\mPP$.

To simplify notation, set $s=\Tr_{k/\qq}\circ\Trd_B$.
For all $[x]_\mPP\in\mJ$ and all $[y]_\mPP\in \Lambda/\mP$, using the claim, we get 

$$\begin{array}{lll}[0]_p&=&\varphi_{\Lambda/\mP,\lambda}([x]_\mPP,[y]_\mPP [h(e)]_\mPP)\cr 
&=&\varphi_{\Lambda/\mP,\lambda}([x]_\mPP, [h(e)]_\mPP[y]_\mPP)\cr 
&=&[s(\lambda\tau(x)h(e)y)]_p=[s(\lambda h(e)\tau(x)y)]_p,\end{array}$$ 
that is $$[0]_p=[s(\lambda \tau(x\tau(h(e)))y)]_p=\varphi_{\Lambda/\mP,\lambda}([x\tau(h(e))]_\mPP,[y]_\mPP)$$ for all $[y]_\mPP\in \Lambda/\mP$. Using the 
nondegeneracy of $\varphi_{\Lambda/\mP,\lambda}$, we get $$[x]_\mPP [\tau(h(e))]_\mPP=[0]_\mPP \mbox { for all }[x]_\mPP\in \mJ,$$ that is 
$$[x]_\mPP \ov{\tau}([h(e)]_\mPP)=[0]_\mPP \mbox { for all }[x]_\mPP\in \mJ.$$
Applying $\ov{\tau}$, we get $$\ov{\tau}([h(e)]_\mPP)[z]_\mPP=[0]_\mPP \mbox { for all }[z]_\mPP\in \ov{\tau}(\mJ).$$
By Lemma \ref{taug}, $\ov{\tau}(\mJ)=[g_{\ov{\tau}}(e)]_\mPP\cdot \Lambda/\mP$.
In particular, we have the equality $[h(e)]_\mPP [g_{\ov{\tau}}(e)]_\mPP=0.$

Applying the isomorphism of Lemma \ref{lem:iso}, we see that $X^n-[\gamma]_\mPP$ is a left divisor $\ov{h}\ov{g}_{\ov{\tau}}$.

By Lemma \ref{lemdeg}, we have $$\deg(\ov{h} \ov{g}_{\ov{\tau}})=\deg(\ov{h})+\deg(\ov{g}_{\ov{\tau}})=\deg(\ov{h})+\deg(\ov{g})=n.$$

Since $\ov{h}\ov{g}_{\ov{\tau}}$ and $X^n-[\gamma]_\mPP$ are both monic, the previous divisibility relation then  implies that $X^n-[\gamma]_\mPP=\ov{h} \  \ov{g}_{\ov{\tau}}.$ The formula for $\ov{g}_{\ov{\tau}}$ follows from Lemma \ref{taug}, taking into account that $\ov{\sigma}$ is trivial.

Since $\mJ\subset\mJ^\perp$, respectively $\mJ=\mJ^\perp$, if and only if $\ov{h}$ is a left divisor of $\ov{g}$, respectively $\ov{h}=\ov{g}$, we get the last part.
\end{proof}

\begin{ex}
Assume that $p\equiv 1 \ [4].$ Let $k=\qq, L=\qq(\sqrt{-p})$ and $\PP=(\sqrt{-p}),$ so we are in the situation of the proposition. Note that $\sigma$ is complex conjugation, and induces the identity morphism on $\OO_L/\PP$. 

Let $\lambda=\dfrac{1}{2}$.  Note that $\mathcal{D}_L=(2\sqrt{-p})$, and thus $$p\lambda^{-1}\mathcal{D}_L^{-1}=(\sqrt{-p})=\PP,$$
so $\lambda$ fulfills all the required conditions.

Now take $\gamma=-1$. Since $\OO_L/\PP\simeq\ff_p$ and $p\equiv 1 \ [4]$, there exists $a\in\zz$ such that $[a]_\PPP^2=[-1]_\PPP$. Let $\ov{g}=X-[a]_\PPP$, Then, $\ov{g}_{\ov{\tau}}=X-[a]_\PPP^{-1}=X+[a]_\PPP$, and we have $$\ov{g}\,\ov{g}_{\ov{\tau}}=X^2-[a]_\PPP^2=X^2+[1]_\PPP=X^2-[\gamma]_\PPP.$$
Hence, the ideal $\mJ$ corresponding to $\ov{g}$ satisfies $\mJ=\mJ^\perp$. In particular, $\mJ\subset\mJ^\perp$. Moreover, 
$\det(\Lambda)=(\dfrac{1}{4}\cdot (4p))^2=p^2.$ Now, $\Lambda$ has rank $4$ over $\zz$, and $\codim_{\ff_p}(\mJ)=1$ by Lemma \ref{lemdeg}. Hence $\det(\Lambda)=p^{4-2\codim_{\ff_p}(\mJ)}$. Theorem \ref{thm-gamma} then predicts that $\Gamma_\mJ$ is an unimodular lattice of rank $4$. Therefore, $\Gamma_\mJ$ is isomorphic to $\zz^4$.

Explicitly, we have $\Gamma_\mJ=\dfrac{1}{\sqrt{p}}(\Lambda (e-a)+\Lambda\sqrt{-p})$.
\end{ex}

\begin{ex}
Assume that $p\equiv 1 \ [4].$ Let $k=\qq, L=\qq(\zeta_p)$, $\mP=(1-\zeta_p),$ and $\gamma=-1$. Picking any generator $\sigma$ of $\Gal(L/\qq)$, $B=(-1,L/\qq,\sigma)$ is a cyclic $\qq$-algebra of degree $p-1$. 

We have $\mathcal{D}_L=\mP^{p-2}$ and $(p)=\mP^{p-1}$, so if we take $\lambda=1$, all the needed assumptions are fulfilled. Moreover, $\Lambda$ has rank $(p-1)^2$ and $\det(\Lambda)=p^{(p-2)(p-1)}$.

Once again, let $a\in\zz$ satisfying $[a]_\PPP^2=[-1]_\PPP$, and set $\ov{g}=X^{\frac{p-1}{2}}-\ov{a}$. Then 
$\ov{g}_{\ov{\tau}}=X^{\frac{p-1}{2}}-[a]_\PPP$ and  $\ov{g} \ \ov{g}_{\ov{\tau}}=X^{p-1}-[\gamma]_\PPP.$

Notice that $\codim_{\ff_p}(\mJ)=\dfrac{p-1}{2}$, and therefore $$(p-1)^2-\codim_{\ff_p}(\mJ)=(p-1)^2-(p-1)=(p-2)(p-1).$$
It follows from Theorem \ref{thm-gamma} that $\Gamma_\mJ$ is unimodular. 

Since $(-\ds\sum_{k=1}^{\frac{p-1}{2}}\zeta_p^k)^*+(-\ds\sum_{k=1}^{\frac{p-1}{2}}\zeta_p^k)=1,$ Lemma \ref{evenlj} shows that $\Gamma_\mJ$ is an even unimodular lattice of rank $(p-1)^2$.

When $p=5$, one may check that we obtain $E_8\perp E_8$.
\end{ex}

\begin{ex}
Assume that $p\equiv 1 \ [4].$ Let $k=\qq(\zeta_p+\zeta_p^{-1}), L=\qq(\zeta_p)$, $\mP=(1-\zeta_p),$ and $\gamma=-1$. Then $\sigma$ is necessarily complex conjugation, and  $B=(-1,L/k,\sigma)$ is a quaternion $k$-algebra.

Once again, we have $\mathcal{D}_L=\mP^{p-2}$ and $(p)=\mP^{p-1}$, so if we take $\lambda=1$, all the needed assumptions are fulfilled. This time, $\Lambda$ is a lattice of rank $4[k:\qq]=2p-2$, and $\det(\Lambda)=p^{2p-4}$.

Let $a\in\zz$ satisfying $[a]_\PPP^2=[-1]_\PPP$, and set $\ov{g}=X-[a]_\PPP$. Then we have
$\ov{g}_{\ov{\tau}}=X^2+[a]_\PPP,$ so  that $\ov{g}\, \ov{g}_{\ov{\tau}}=X^2-[\gamma]_\PPP.$

Notice that $\codim_{\ff_p}(\mJ)=1$, and therefore $$2p-2-\codim_{\ff_p}(\mJ)=2p-4.$$
It follows from Theorem \ref{thm-gamma} that $\Gamma_\mJ$ is unimodular. 

As before, $\Gamma_\mJ$ is in fact an even unimodular lattice of rank $2p-2$.

Thus, when $p=5$, we obtain $E_8$. When $p=13$, one may check that we obtain $A_{13}\perp A_{13}$.
\end{ex}

\begin{ex}
Assume that $p$ is a prime number satisfying $p \equiv 1 \ [8]$. Let $k=\qq(\sqrt{p}), L=\qq(\zeta_p+\zeta_p^{-1}),$ and $\gamma=-1$, so $B=(-1,L/k,\sigma)$ has degree $\dfrac{p-1}{4}$.

Let $\mP$ the unique prime ideal of $\OO_L$ lying above $p$. We have $(p)=\mP^{\frac{p-1}{2}}$ and $\mathcal{D}_L=\mP^{\frac{p-3}{2}}$, so $\lambda=1$ fulfills all the needed assumptions. Here, $\Lambda$ has rank $2\Bigl(\dfrac{p-1}{4}\Bigr)^2$, and determinant $p^{\frac{(p-3)(p-1)}{8}}$.

Let $a\in\zz$ such that $[a]^2_\mPP=[-1]_\mPP,$ and set $\ov{g}=X^{\frac{p-1}{8}}-[a]_\PPP$. As before, $\ov{g} \,\ov{g}_{\ov{\tau}}=X^{\frac{p-1}{4}}-[\gamma]_\PPP$.

We have $\codim_{\ff_p}(\mJ)=\dfrac{p-1}{8}$, and thus 
$$2\Bigl(\dfrac{p-1}{4}\Bigr)^2-2\codim_{\ff_p}(\mJ)=\dfrac{(p-1)^2}{8}-\dfrac{p-1}{4}=\dfrac{p^2-4p+3}{8}=\dfrac{(p-3)(p-1)}{8}.$$
Theorem \ref{thm-gamma} then implies that $\Gamma_\mJ$ is an unimodular lattice of rank $2\Bigl(\dfrac{p-1}{4}\Bigr)^2.$

When $p=17$, we get a lattice isomorphic to $\zz^{32}$.
\end{ex}

We now prove a result of similar flavor, but with different assumptions.

Recall that for any $\lambda\in\ L^\times$ such that $v_\PPP(\lambda)=0$, we may define a residue class $[\lambda]_\PPP\in(\OO_L/\PP)^\times$.

Indeed, we may write $(\lambda)=\mathfrak{A}\mathfrak{B}^{-1}$, where $\mathfrak{A}$ and $\mathfrak{B}$ are nonzero ideals of $\OO_L$ coprime to $\PP$. 
In particular, $\lambda\mathfrak{B}=\mathfrak{A}$. Picking an element $b\in\mathfrak{B}$, we may then write $\lambda=\dfrac{a}{b}$, where $a,b\in\OO_L\setminus \PP$. It is easy to check that the class $[a]_\PPP[b]_\PPP^{-1}\in\OO_L/\PP$ does not depend on the choice of $a$ and $b$, so we may set $$[\lambda]_\PPP=[a]_\PPP[b]_\PPP^{-1}.$$

Notice that we can even take $b\in\zz$ if necessary.

 We then have the following result.

\begin{prop}
Assume that $p$ is unramified or tamely ramifies in $L$. Let $\pp$ be a prime ideal of $\OO_K$ lying above $p$ such that $\pp^*=\pp$, and which is inert in $L$, and let $\PP=\pp\OO_L$. Let $\bar{g}=\ds\sum_{i=0}^dX^i [a_i]_\PPP$ be a monic divisor of $X^n-[\gamma]_\PPP,$ and let $\mJ$ be the corresponding left ideal of $\Lambda/\mP$.

Then $v_\PPP(\lambda)=0$, and $\mJ^\perp$ corresponds to the monic divisor 
$$\bar{h}=\frac{X^n-[\gamma]_\PPP}{ \ov{g}_{\ov{\tau},\lambda}},$$
where $ \ov{g}_{\ov{\tau},\lambda}=\ds\sum_{k=0}^dX^k[\lambda a^*_{d-k}]_\PPP^{\ov{\sigma}^k} [\lambda a_0^{*-1}]_\PPP^{\ov{\sigma}^d}.$

In particular, $\mJ$ is self-orthogonal if and only if  $\ov{g} \ \ov{g}_{\ov{\tau},\lambda}$ is a multiple of $X^n-[\gamma]_\PPP$,
and self-dual if and only if $\ov{g} \,  \ov{g}_{\ov{\tau},\lambda}=X^n-[\gamma]_\PPP.$
\end{prop}

\begin{proof}
Since $\varphi_{\Lambda/\mP,\lambda}$ is well-defined, we have $\lambda\PP\subset p\mathcal{D}_L^{-1}$. In particular, $v_\PPP(\lambda)+1\geq e_p-d_p$, that is $e_p-d_p-v_\PPP(\lambda)\leq 1$.

Since $\varphi_{\Lambda/\mP,\lambda}$ is nondegenerate, we have $p(\lambda\mathcal{D}_L^{-1})\cap \OO_L=\PP$. Taking  $\PP$-adic valuation on both sides gives $$\max(e_p-d_p-v_\PPP(\lambda),0)=1.$$
The inequality above shows that this is possible only if $e_p-d_p-v_\PPP(\lambda)=1$. Since $p$ is unramified or tamely ramifies in $L$, we have $e_p-d_p=1$ in both cases, and thus $v_\PPP(\lambda)=0$.

To simplify notation, set $s=\Tr_{k/\qq}\circ\Trd_B$.

Write $\lambda=\dfrac{a}{b}$, with $a\in\OO_L\setminus\{0\}$, and $b\in\zz\setminus\{0\}$.

Keeping the notation of the proposition, we have $\mJ=\Lambda/\mP \cdot [g(e)]_\mPP$, and $\mJ^\perp=\Lambda/\mP \cdot [h(e)]_\mPP$.

For all $\ov{x}\in\mJ$ and all $[y]_\mPP\in \Lambda/\mP$,  we get 

$$[0]_p=\varphi_{\Lambda/\mP,\lambda}([x]_\mPP,[y]_\mPP [h(e)]_\mPP)=[s(\lambda\tau(x)yh(e))]_p.$$

Let $a'\in\OO_L\setminus\{0\}$ such that $[a']_\PPP [a]_\PPP=\ov{1}\in\OO_L/\PP$. Then $1-a'a\in\PP$, and thus $\tau(x)y(1-a'a)h(e)\in\mP$, since $\mP$ is a two-sided ideal. Since $1=a'a+(1-a'a),$ Lemma \ref{lem-trmp} then yields
$$[0]_p=[s(\lambda\tau(x)ya'ah(e)))]_p=[s(ah(e)\lambda\tau(x)ya')]_p.$$
Since $b^{-1}$ commutes with every element of $\Lambda$, we get 
$$\begin{array}{lll}[0]_p&=&[s(\lambda h(e)a\tau(x)ya')]_p\cr &=&[s(\lambda \tau(x \tau(a) \tau(h(e))y a']_p\cr 
&=&\varphi_{\Lambda/\mP,\lambda}([x \tau(a) \tau(h(e))]_\mPP,[y]_\mPP[a']_\mPP)\end{array}$$ for all $[x]_\mPP\in\mJ$ and all $[y]_\mPP\in \Lambda/\mP$.
Since $[a']_\mPP$ is invertible in $\OO_L/\PP$, and thus in $\Lambda/\mP,$ this is equivalent to $$[0]_p=\varphi_{\Lambda/\mP,\lambda}([x \tau(a) \tau(h(e)]_\mPP,[y]_\mPP)$$ for all $[x]_\mPP\in\mJ$ and all $[y]_\mPP\in \Lambda/\mP$. Since $\varphi_{\Lambda/\mP,\lambda}$ is nondegenerate, this yields $[x \tau(a) \tau(h(e)]_\mPP=[0]_\mPP$. Applying $\ov{\tau}$, we get $$[h(e)a ]_\mPP\ov{\tau}([x]_\mPP)=[0]_\mPP$$ for all $\ov{x}\in\mJ$, that is  $[h(e)a]_\mPP[z]_\mPP=[0]_\mPP$ for all $[z]_\mPP\in\ov{\tau}(\mJ)$. By Lemma \ref{taug}, this right ideal is generated by 
$[g_{\ov{\tau}}(e)]_\mPP$, so we have in particular $$[h(e)]_\mPP[a]_\mPP[g_{\ov{\tau}}(e)]_\mPP=[0]_\mPP.$$
Multiplying by $[b]_\mPP^{-1}$ then yields
 $$[h(e)]_\mPP[\lambda]_\mPP[g_{\ov{\tau}}(e)]_\mPP=[0]_\mPP,$$ taking into account that $[b]_\mPP^{-1}$ is  central in $\Lambda/\mP$.

Notice now that $ [\lambda]_\PPP \ov{g}_{\ov{\tau}}=\ds\sum_{k=0}^dX^k[\lambda a^*_{d-k}]_\PPP^{\ov{\sigma}^k}$.
 Multiplying by a suitable element of $\OO_L/\PP$ finally yields $$[h(e)]_\mPP [g_{\ov{\tau},\lambda}(e)]_\mPP=[0]_\mPP.$$

Now, we may argue as in the proof of Proposition \ref{proplram} to get the desired result.
\end{proof}

\begin{ex}
Assume that $p\equiv 3 \ [8]$. Let $k=\qq(\sqrt{-p})$ and $L=\qq(\sqrt{-p},\sqrt{2})$. Since $\qq(\sqrt{-p})$ and $\qq(\sqrt{2})$ have discriminants $-p$ and $8$ respectively, they are arithmetically disjoint over $\qq$. Hence , $\mathcal{D}_L=(2\sqrt{2}\sqrt{-p}),$ and one may check that $\pp=\sqrt{-p}\OO_k$ is inert in $L$ (use Proposition \ref{proph2} and the fact that $2$ is not a square modulo $p$). 

The unique nontrivial $k$-automorphism $\sigma$ of $L$ sends $\sqrt{2}$ to $-\sqrt{2}$ and $\sqrt{-p}$ to $\sqrt{-p}.$

Set $\lambda=\dfrac{1+\sqrt{2}}{2\sqrt{2}}$. The Galois group of $L/\qq$ is generated by $\sigma$ and complex conjugation. It is then easy to check that $\lambda$ is a totally positive element of $L$. Moreover, we have $$p\lambda^{-1}\mathcal{D}_L^{-1}=\dfrac{\sqrt{-p}}{1+\sqrt{2}}\OO_L=\sqrt{-p}\OO_L=\PP,$$ since $1+\sqrt{2}$ is a unit of $\OO_L$.

Notice finally that $\omega=[\sqrt{2}]_\PP\in \OO_L/\PP$ satisfies $\omega^2=\ov{2}$. The assumption on $p$ implies that $2$ is not a square modulo $p$, hence $([1]_\PPP,\omega)$ is an $\ff_p$-basis of $\OO_L/\PP$. Complex conjugation becomes trivial on $\OO_L$, while $\ov{\sigma}$ is $\ff_p$-linear and send $\omega$ to $-\omega$.

Since $-1$ is not a square modulo $p$ either, so $-2$ is a square modulo $p$. Let $u\in\zz$ such that $[u]_\PPP^2=[-2]_\PPP$. 

Take $\gamma=-1$, and let $\alpha=[u]_\PPP+[u]_\PPP^{-1}\omega$. Notice that 
$$\alpha\alpha^{\ov{\sigma}}=([u]_\PPP+[u]_\PPP^{-1}\omega)([u]_\PPP-[u]_\PPP^{-1}\omega)=[-2]_\PPP-[-2]_\PPP^{-1}[2]_\PPP=[-1]_\PPP=[\gamma]_\PPP.$$
It follows that we have $$(X-\alpha)(X+\alpha^{\ov{\sigma}})=X^2-[\gamma]_\PPP,$$ so that 

$\ov{g}=X-\alpha$ is a divisor of $X^2-[\gamma]_\PPP$. We have $[\lambda]_\PPP=\dfrac{[1]_\PPP+\omega}{[2]_\PPP\omega}$, and thus $$\dfrac{[\lambda]_\PPP}{[\lambda]_\PPP^{\ov{\sigma}}}=\dfrac{[1]_\PPP+\omega}{\omega-[1]_\PPP}=(\omega+[1]_\PPP)^2=[3]_\PPP+[2]_\PPP\omega.$$

Now $(\alpha^*)^{\ov{\sigma}}=[u]_\PPP-[u]_\PPP^{-1}\omega$, so that $$\alpha^{\ov{\sigma}}(\alpha^*)^{\ov{\sigma}}=((\alpha)^{\ov{\sigma}})^2=([u]_\PPP+[u]_\PPP^{-1}\omega)^2=-[3]_\PPP-[2]_\PPP\omega=-\dfrac{[\lambda]_\PPP}{[\lambda]_\PPP^{\ov{\sigma}}}.$$
Consequently, $\dfrac{[\lambda]_\PPP}{[\lambda]_\PPP^{\ov{\sigma}}(\alpha^*)^{\ov{\sigma}}}=-(\alpha)^{\ov{\sigma}}$, and 
thus $\ov{g}_{\ov{\tau},\lambda}=X+\alpha^{\ov{\sigma}}.$ It follows that $\ov{g}\ov{g}_{\ov{\tau},\lambda}=X^2-[\gamma]_\PPP.$

Here, $\codim_{\ff_p}(\mJ)=2$, and $\Lambda$ has rank $8$ over $\zz$, where $\mJ$ is the left ideal corresponding to $\ov{g}$. Now, $$\det(\Lambda)=(\dfrac{1}{64}\cdot (64p^2))^2=p^4=p^{8-2\codim_{\ff_p}(\mJ)},$$ and Theorem \ref{thm-gamma} says that $\Gamma_\mJ$ is an integral unimodular lattice of rank $8$. Since $\Bigl(\dfrac{1-\sqrt{-p}}{2}\Bigr)+\Bigl(\dfrac{1-\sqrt{-p}}{2}\Bigr)^*=1$, Lemma \ref{evenlj} ensures that $\Gamma_\mJ$ is even. Hence $\Gamma_\mJ$ is isomorphic to $E_8$.
\end{ex}

%
%

\end{document}